\documentclass[openacc]{rsproca_new}

\newtheorem{thm}{\bf Theorem}[section]
\newtheorem{prop}{\bf Proposition}[section]

\usepackage{mathrsfs}

\usepackage{bm}

\usepackage{titlesec}

\usepackage{multirow}

\usepackage[titletoc,title]{appendix}

\newcommand{\ds}{\displaystyle}

\def\EXP{\textrm{{\large e}}}

\newcommand{\lag}{{\mathcal \L}}
\newcommand{\ol}{\overline{\lag}}
\newcommand{\lam}{{\Lambda}}
\newcommand{\olam}{\overline{\lam}}
\newcommand{\lagh}{\hat{\lag}}
\newcommand{\olh}{\hat{\ol}}

\newcommand{\olamh}{\hat{\olamh}}

\newcommand{\al}{{\bm{\alpha}}}
\newcommand{\bt}{{\bm{\beta}}}
\newcommand{\gm}{{\bm{\gamma}}}

\newcommand{\ii}{\mathsf{i}}

\newcommand{\ccy}{x_a}
\newcommand{\ccya}{z_w}
\newcommand{\ccyb}{x}
\newcommand{\ccyc}{z_n}
\newcommand{\ccz}{y_d}
\newcommand{\ccza}{z_e}
\newcommand{\cczb}{y}
\newcommand{\cczc}{z_s}
\newcommand{\ccx}{x}
\newcommand{\cca}{y_c}
\newcommand{\ccb}{x_c}
\newcommand{\ccc}{x_d}
\newcommand{\ccd}{x_b}
\newcommand{\cce}{y_b}
\newcommand{\ccf}{y_a}
\newcommand{\ccpc}{\alpha}
\newcommand{\ccpa}{\alpha_1}
\newcommand{\ccpb}{\alpha_2}
\newcommand{\ccqc}{\beta}
\newcommand{\ccqa}{\beta_1}
\newcommand{\ccqb}{\beta_2}
\newcommand{\ccra}{\gamma_1}
\newcommand{\ccrb}{\gamma_2}
\newcommand{\ccpp}{{\al}}
\newcommand{\ccqq}{{\bt}}
\newcommand{\ccrr}{{\gm}}

\renewcommand{\L}{L}
\newcommand{\LL}{L_1}
\newcommand{\M}{L_2}
\newcommand{\Lb}{L_4}
\newcommand{\Mb}{L_3}

\newcommand{\LA}{L_A}
\newcommand{\LC}{L_C}
\newcommand{\LCi}{\hat{L}_C}

\newcommand{\Lxa}{\L_{x^2}}
\newcommand{\Lxb}{\L_{x^1}}
\newcommand{\Lxc}{\L_{x^0}}
\newcommand{\Dxa}{\Delta_{x^2}}
\newcommand{\Dxb}{\Delta_{x^1}}
\newcommand{\Dxc}{\Delta_{x^0}}

\newcommand{\Af}[7]{\overline{A}(#1;#2,#3,#4,#5;#6,#7)}
\newcommand{\A}[7]{A(#1;#2,#3,#4,#5;#6,#7)}

\newcommand{\C}[7]{C(#1;#2,#3,#4,#5;#6,#7)}
\newcommand{\Cf}[7]{\overline{C}(#1;#2,#3,#4,#5;#6,#7)}

\newcommand{\Aft}[7]{\overline{a}(#1;#2,#3,#4,#5;#6,#7)}
\newcommand{\At}[7]{a(#1;#2,#3,#4,#5;#6,#7)}

\newcommand{\Ct}[7]{c(#1;#2,#3,#4,#5;#6,#7)}
\newcommand{\Cft}[7]{\overline{c}(#1;#2,#3,#4,#5;#6,#7)}

\begin{document}

\title{Discrete integrable equations on face-centered cubics: consistency and Lax pairs of corner equations}

\author{
Andrew P.~Kels}

\address{Scuola Internazionale Superiore di Studi Avanzati,\\ Via Bonomea 265, 34136 Trieste, Italy}

\subject{mathematical physics}

\keywords{multidimensional consistency, discrete integrability, Lax pairs, lattice equations}

\corres{Andrew P.~Kels\\
\email{andrew.p.kels@gmail.com}}

\begin{abstract}
A new set of discrete integrable equations, called face-centered quad equations, was recently obtained using new types of interaction-round-a-face solutions of the classical Yang-Baxter equation.  These equations satisfy a new formulation of multidimensional consistency, known as consistency-around-a-face-centered-cube (CAFCC), which requires consistency of an overdetermined system of fourteen five-point equations on the face-centered cubic unit cell.  In this paper a new formulation of CAFCC is introduced where so-called type-C equations are centered at faces of the face-centered cubic unit cell, whereas previously they were only centered at corners.  This allows type-C equations to be regarded as independent multidimensionally consistent integrable systems on higher-dimensional lattices and is used to establish their Lax pairs.
\end{abstract}


\begin{fmtext}
\section{Introduction}

The property of multidimensional consistency has been at the forefront of studies of integrability of partial difference equations \cite{hietarinta_joshi_nijhoff_2016} for the past two decades.  For partial difference equations defined on faces of a square lattice, commonly known as quad equations, multidimensional consistency arises as a consequence of the consistency of an overdetermined system of six equations on the faces of a cube, a property commonly known as consistency-around-a-cube (CAC) \cite{DoliwaSantini,nijhoffwalker,BobSurQuadGraphs,ABS}.  Furthermore, by comparing two different evolutions on the cube one may derive Lax pairs for CAC equations \cite{NijhoffQ4Lax,BobSurQuadGraphs,BHQKLax,HietarintaNEWCAC}, through which the equations themselves are essentially reinterpreted as their own Lax matrices.

The author has recently introduced the concept of face-centered quad equations \cite{Kels:2020zjn} as a new type of multidimensionally consistent equation on the square  \phantom{equations} 

\end{fmtext}


\maketitle

\noindent
lattice.  The analogue of CAC for face-centered quad equations is the property of consistency-around-a-face-centered-cube (CAFCC), which requires consistency of an overdetermined system of fourteen equations on the face-centered cube (face-centered cubic unit cell).  Similarly to CAC equations, a comparison of two different directions of evolution on the face-centered cube was shown to lead to Lax pairs \cite{KelsLax}, through which the equations themselves are reinterpreted as their own Lax matrices.  Besides multidimensional consistency and Lax pairs, the CAFCC equations also possess vanishing algebraic entropy \cite{GubbiottiKels} (as do CAC equations), which is another important property associated to integrability of lattice equations \cite{BellonViallet1999,Tremblay2001,Viallet2006}.  

In the original formulation of CAFCC \cite{Kels:2020zjn} there arose three types of equations which were denoted as one of type-A, -B, or -C.  The type-C equations were centered at corners of the face-centered cube and provided an intermediate evolution equation between type-A and -B equations that were centered at faces of the face-centered cube.  This means that CAFCC in its original formulation can only imply the consistency of systems of either type-A or type-B equations in higher-dimensional lattices.  Furthermore, the method that was introduced for deriving Lax pairs \cite{KelsLax} only applies to CAFCC equations appearing on faces of the face-centered cube, and thus Lax pairs have only been obtained for type-A and type-B equations, but not for type-C equations.

The purpose of this paper is to establish a new formulation of CAFCC where the type-C equations are centered at faces of the face-centered cube. The motivation behind the search for this new formulation of CAFCC arises from the recent discovery of the vanishing algebraic entropy for type-C equations in appropriate arrangements in the square lattice \cite{GubbiottiKels}, suggesting that they are integrable systems which should also have consistent extensions into higher-dimensional lattices, similarly to type-A and -B equations. The new formulation of CAFCC of this paper gives the multidimensional consistency of systems of type-C equations in combination with type-A equations, and furthermore allows for the derivation of their Lax pairs.

This paper is organised as follows.  In \S\ref{sec:FCQECAFCC} an overview of the concept of face-centered quad equations is given, and the main result of this section is the new formulation of CAFCC involving only type-A and type-C equations which is given in \S\ref{sec:FCQECAFCC}\ref{sec:CAFCC}.  In \S\ref{sec:Lax} it is shown how the property of CAFCC leads to Lax pairs of type-C equations, and explicit examples are presented.  The type-A and type-C equations obtained from \cite{Kels:2020zjn} are given in Appendix \ref{app:equations}, and in Appendix \ref{sec:YBE} a new form of the classical Yang-Baxter equation is given that is related to the property of CAFCC of \S\ref{sec:FCQECAFCC}\ref{sec:CAFCC}.

\section{Face-centered quad equations and CAFCC}\label{sec:FCQECAFCC}

The author has recently introduced the concept of face-centered quad equations and their consistency as CAFCC \cite{Kels:2020zjn}. A face-centered quad equation may be written as 
\begin{align}
\A{x}{x_a}{x_b}{x_c}{x_d}{\ccpp}{\ccqq}=0,
\end{align}
where $A$ is a multivariate polynomial of five variables $x,x_a,x_b,x_c,x_d$, with linear dependence on each of the four variables $x_a,x_b,x_c,x_d$, but no restriction on the degree of $x$.  This equation is associated to a square as shown on the left of Figure \ref{fig:3fig4quad}, where the variables $x_a,x_b,x_c,x_d$ are at corners, and the variable $x$ is centered at the face.  The multilinearity in the four corner variables $x_a,x_b,x_c,x_d$ is a natural requirement to have a well-defined evolution in the square lattice.  The parameters $\ccpp$ and $\ccqq$ each have two components
\begin{align}
\label{pardefs}
\ccpp=(\ccpa,\ccpb),\qquad\ccqq=(\ccqa,\ccqb),
\end{align}
and are associated to horizontal and vertical lines as indicated in Figure \ref{fig:3fig4quad}.

\begin{figure}[h!]
\centering
\includegraphics{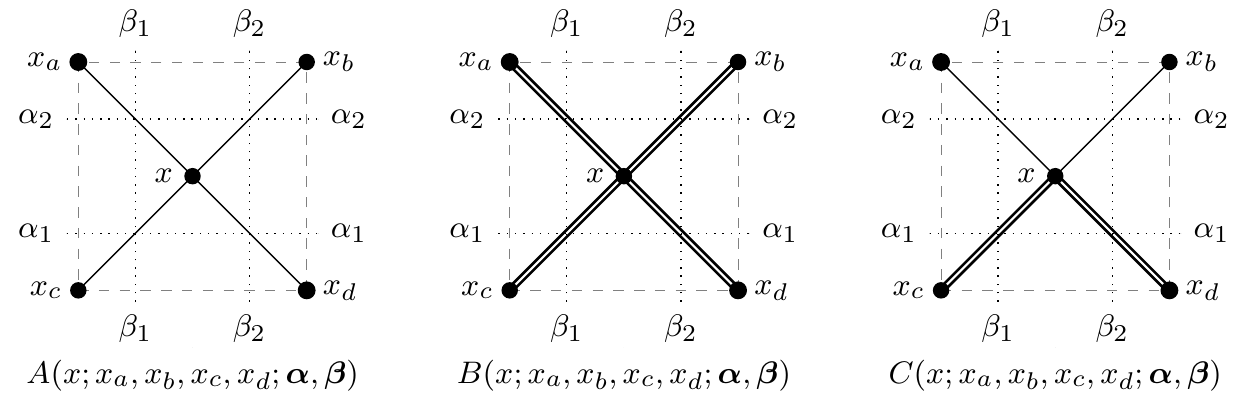}
\caption{From left to right, a type-A, type-B, and type-C face-centered quad equation.}
\label{fig:3fig4quad}
\end{figure}

A general polynomial form of a face-centered quad equation is given by
\begin{multline}\label{afflin}
\A{\ccx}{x_a}{x_b}{x_c}{x_d}{\ccpp}{\ccqq} 
=\kappa_1x_ax_bx_cx_d+\kappa_2x_ax_bx_c+\kappa_3x_ax_bx_d+\kappa_4x_ax_cx_d+\kappa_5x_bx_cx_d
\\
\begin{aligned}
+\kappa_6x_ax_b+\kappa_7x_ax_c+\kappa_8x_ax_d+\kappa_9x_bx_c+\kappa_{10}x_bx_d+\kappa_{11}x_cx_d \\
+\kappa_{12}x_a+\kappa_{13}x_b+\kappa_{14}x_c+\kappa_{15}x_d+\kappa_{16}=0,
\end{aligned}
\end{multline}
where the coefficients $\kappa_i=\kappa_i(x,\al,\bt)$ ($i=1,\ldots,16$) depend on the face variable $\ccx$ and the four components $\ccpa,\ccpb,\ccqa,\ccqb$ of the parameters $\ccpp,\ccqq$.  The multilinear expression \eqref{afflin} resembles the general expression for regular quad equations, except in the latter case the coefficients would depend only on two one-component parameters and no variables.

Originally, three types of face-centered quad equations were distinguished as type-A, -B, or -C.  The type-A equations satisfy the following three symmetries
\begin{align}\label{refsym}
    \A{x}{x_a}{x_b}{x_c}{x_d}{\al}{\bt}=-\A{x}{x_b}{x_a}{x_d}{x_c}{\al}{\hat{\bt}},
\\
\label{symAB}
   \A{x}{x_a}{x_b}{x_c}{x_d}{\al}{\bt}=-\A{x}{x_c}{x_d}{x_a}{x_b}{\hat{\al}}{\bt},
\\
\label{symA}
    \A{x}{x_a}{x_b}{x_c}{x_d}{\al}{\bt}=-\A{x}{x_d}{x_b}{x_c}{x_a}{\bt}{\al},
\end{align}
where $\hat{\al}$ and $\hat{\bt}$ represent $\al$ and $\bt$ with the components exchanged
\begin{align}
    \hat{\al}=(\alpha_2,\alpha_1),\qquad\hat{\bt}=(\beta_2,\beta_1).
\end{align}
Type-B equations only satisfy \eqref{refsym} and \eqref{symAB}, while type-C equations only satisfy \eqref{refsym}.  The symmetries \eqref{refsym}--\eqref{symA} are natural analogues for face-centered quad equations of the square symmetries that are satisfied by regular quad equations.  

For the purposes of this paper, it will be useful to distinguish the three different types of face-centered quad equations by using different types of edges as is shown in Figure \ref{fig:3fig4quad}.  For face-centered quad equations that were obtained in \cite{Kels:2020zjn}, this graphical representation is related to the existence of ``four-leg'' forms of the equations, which arose as a step in the derivation of the equations \eqref{afflin} that came from the classical Yang-Baxter equation (however, a generic face-centered quad equation \eqref{afflin} does not have this property).  Specifically, the face-centered quad equations derived in \cite{Kels:2020zjn} may be written in the following forms
\begin{align}
\label{4leg}
&\textrm{Type-A:}\;&\frac{
a(x;x_a;\alpha_2,\beta_1)a(x;x_d;\alpha_1,\beta_2)}{a(x;x_b;\alpha_2,\beta_2)a(x;x_c;\alpha_1,\beta_1)}=1, \\
%
\label{4legb}
&\textrm{Type-B:}\;&\frac{b(x;x_a;\alpha_2,\beta_1)b(x;x_d;\alpha_1,\beta_2)}{b(x;x_b;\alpha_2,\beta_2)b(x;x_c;\alpha_1,\beta_1)}=1, \\
\label{4legc}
&\textrm{Type-C:}\;&\frac{a(x;x_a;\alpha_2,\beta_1)c(x;x_d;\alpha_1,\beta_2)}{a(x;x_b;\alpha_2,\beta_2)c(x;x_c;\alpha_1,\beta_1)}=1.
\end{align}
The function $a(x,y;\ccpc,\ccqc)$ satisfies $a(x,y;\ccpc,\ccqc)a(x,y;\ccqc,\ccpc)=1$, and the three functions $a(x;y;\ccpc,\ccqc)$, $b(x;y;\ccpc,\ccqc)$, and $c(x;y;\ccpc,\ccqc)$, are each fractional linear functions of a corner variable $y$.  The function $a(x;y;\ccpc,\ccqc)$ is associated to the solid edges of both type-A and type-C equations in Figure \ref{fig:3fig4quad}, and the functions $b(x;y;\ccpc,\ccqc)$ and $c(x;y;\ccpc,\ccqc)$ are associated to the solid double edges of type-B and type-C equations respectively.  The reason that the edges of both type-A and type-B equations appear for type-C equations, is that the latter originally arose in CAFCC as an intermediate equation between the ``face'' vertices of both type-A and type-B equations.  There are type-B and type-C equations for which the functions $b(x;y;\ccpc,\ccqc)$ and $c(x;y;\ccpc,\ccqc)$ are the same, but to satisfy CAFCC \cite{Kels:2020zjn} it requires for some cases that the functions are different, which is the reason that they are denoted as different functions.  Examples of the functions  $a(x;y;\ccpc,\ccqc)$ and  $c(x;y;\ccpc,\ccqc)$ are listed in Table \ref{table-BC2} in Appendix \ref{app:equations} for type-A and type-C equations that will be used in this paper.

Starting from a four-leg form, the multilinear form \eqref{afflin} of a face-centered quad equation is typically recovered by multiplying both sides of one of \eqref{4leg}--\eqref{4legc} by its denominator, and cancelling common terms appearing on both sides.  On the other hand, a generic multilinear face-centered quad equation \eqref{afflin} does not possess a four-leg form. 
The explicit multilinear face-centered quad equations and associated four-leg forms for type-A and type-C equations that were derived in \cite{Kels:2020zjn} are given in Appendix \ref{app:equations}.  The expressions for the type-A equations in the four-leg form \eqref{4leg} are related to certain expressions for discrete Toda equations that have been known for a long time \cite{AdlerPlanarGraphs,BobSurQuadGraphs,Suris_DiscreteTimeToda}, while the type-B and -C equations in the four-leg form \eqref{4legb} and \eqref{4legc} (as well as multilinear form) appear to be new equations that arose from the discovery of CAFCC \cite{Kels:2020zjn}.

\subsection{A new formulation of CAFCC}\label{sec:CAFCC}

Type-A equations have been shown to satisfy CAFCC either by themselves or in combination with type-B and type-C equations \cite{Kels:2020zjn}.  In either case, both type-A and type-B equations only appeared centered at faces of the face-centered cube, and type-C equations only appeared centered at corners.  Here, a new formulation of CAFCC will be given that involves both type-A and type-C equations that are centered at both faces and corners.

Let two different type-A equations, and two different type-C equations, be denoted as
\begin{align}\label{typeACAFCC}
\A{x}{x_a}{x_b}{x_c}{x_d}{\al}{\bt}=0,\qquad\Af{x}{x_a}{x_b}{x_c}{x_d}{\al}{\bt}=0,
\\
\label{typeCCAFCC}
\C{x}{x_a}{x_b}{x_c}{x_d}{\al}{\bt}=0,\qquad\Cf{x}{x_a}{x_b}{x_c}{x_d}{\al}{\bt}=0.
\end{align}
For convenience, the two different type-C equations \eqref{typeCCAFCC} will be distinguished graphically by different orientations of directed edges, as shown in Figure \ref{fig:2typeC}.

\begin{figure}[h!]
\centering
\includegraphics{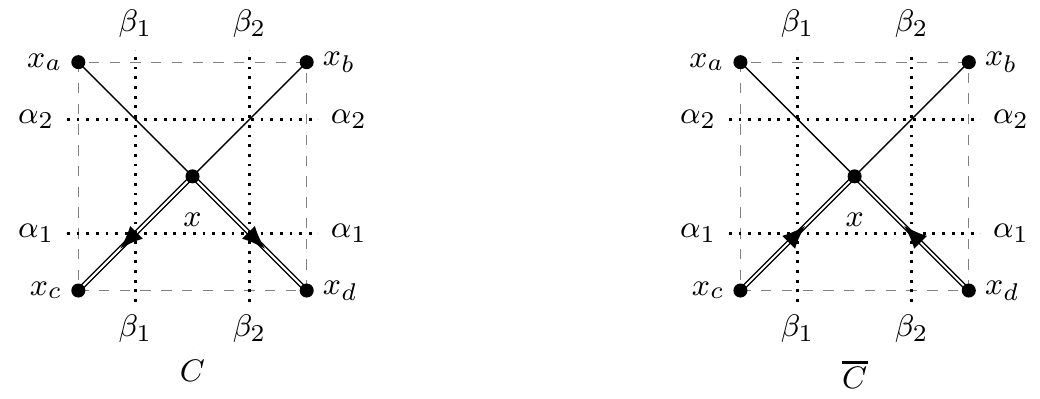}
\caption{Two different type-C equations \eqref{typeCCAFCC} distinguished by the orientations of directed edges.  The convention used in the following is that $C$ will be assigned to faces of the face-centered cube, and $\overline{C}$ will be assigned to corners (see the top diagram of Figure \ref{fig:CAFCCcube}).}
\label{fig:2typeC}
\end{figure}


Two halves of the face-centered cube are shown in the top diagram of Figure \ref{fig:CAFCCcube}. This presentation is used to more clearly show how edges and parameters are arranged, while it also reflects the form of Yang-Baxter equation that is associated to CAFCC (see Appendix \ref{sec:YBE}).  As usual, the three parameters $\al$, $\bt$, $\gm$, are associated to three orthogonal lattice directions, and the components of $\gm$ are exchanged on two faces which meet orthogonally.

There are fourteen equations on the face-centered cube shown in the top diagram of Figure \ref{fig:CAFCCcube}, made up of the six equations
\begin{align}\label{6face}
\begin{array}{rrr}
\C{\ccyb}{\ccy}{\ccd}{\ccb}{\ccc}{\ccpp}{\ccqq}=0, 
&\C{\ccya}{\ccf}{\ccy}{\cca}{\ccb}{\ccpp}{\ccrr}=0, 
&\A{\ccyc}{\ccf}{\cce}{\ccy}{\ccd}{\ccrr}{\ccqq}=0, \\[0.1cm]
\C{\cczb}{\ccf}{\cce}{\cca}{\ccz}{\ccpp}{\ccqq}=0, 
&\C{\ccza}{\cce}{\ccd}{\ccz}{\ccc}{\ccpp}{\ccrr}=0, 
&\Af{\cczc}{\cca}{\ccz}{\ccb}{\ccc}{\ccrr}{\ccqq}=0,
\end{array}
\end{align}
on faces, and the eight equations
\begin{align}\label{8corner}
\begin{array}{rr}
\ds\A{\ccy}{\ccya}{\ccf}{\ccyb}{\ccyc}{(\ccqa,\ccrb)}{(\ccpb,\ccra)}=0, \quad
&\A{\ccd}{\ccza}{\cce}{\ccyb}{\ccyc}{(\ccqb,\ccrb)}{(\ccpb,\ccra)}=0, \\[0.1cm]
\Cf{\ccb}{\cczc}{\cca}{\ccyb}{\ccya}{(\ccpa,\ccra)}{(\ccqa,\ccrb)}=0, \quad
&\Cf{\ccc}{\ccz}{\cczc}{\ccza}{\ccyb}{(\ccpa,\ccra)}{(\ccrb,\ccqb)}=0, \\[0.1cm]
\A{\ccf}{\ccya}{\ccy}{\cczb}{\ccyc}{(\ccqa,\ccra)}{(\ccpb,\ccrb)}=0, \quad
&\A{\cce}{\ccza}{\ccd}{\cczb}{\ccyc}{(\ccqb,\ccra)}{(\ccpb,\ccrb)}=0, \\[0.1cm]
\Cf{\cca}{\ccb}{\cczc}{\ccya}{\cczb}{(\ccpa,\ccrb)}{(\ccra,\ccqa)}=0, \quad
&\Cf{\ccz}{\ccc}{\cczc}{\ccza}{\cczb}{(\ccpa,\ccrb)}{(\ccra,\ccqb)}=0, 
\end{array}
\end{align}
on corners. An important difference from the original formulation of CAFCC is that there are now type-C equations in \eqref{6face} that are centered at faces of the face-centered cube.

\begin{figure}[h!]
\centering
\includegraphics{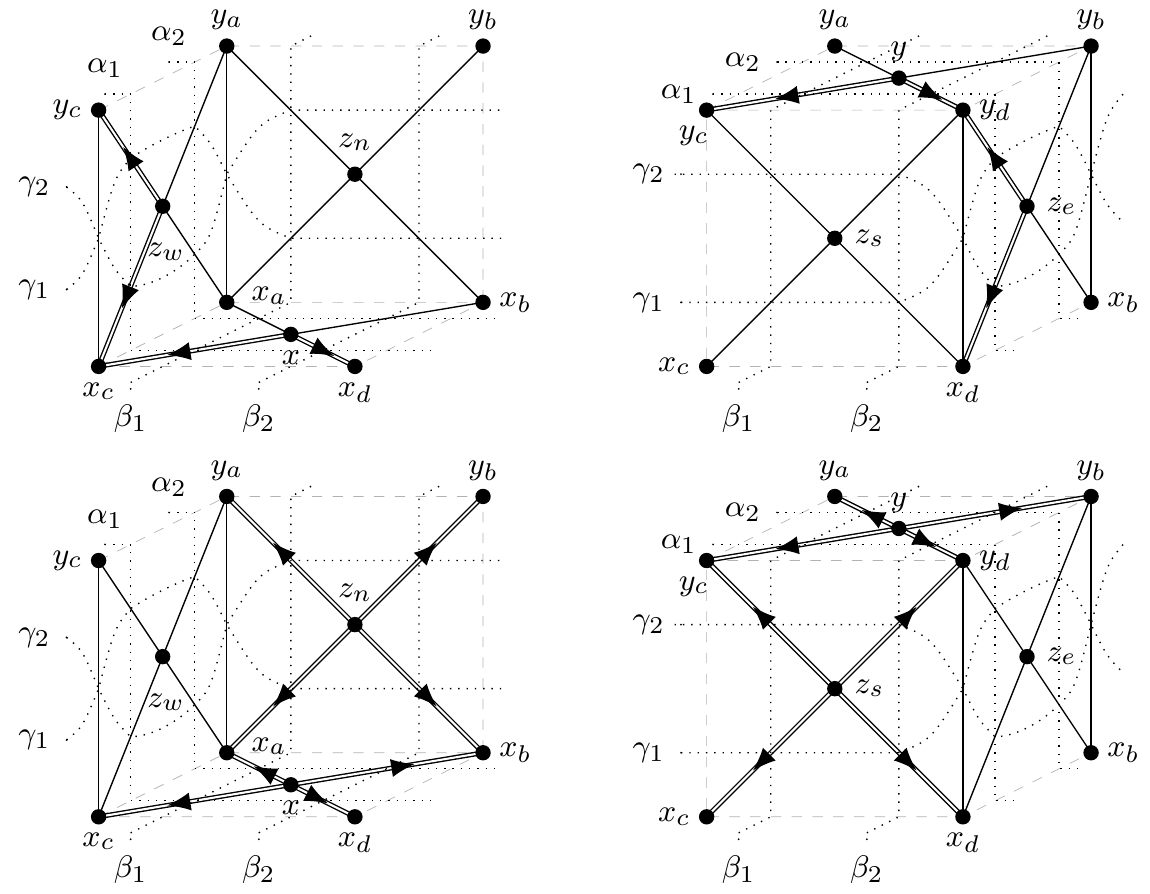}
\caption{Top: Labelling of vertices and edges for CAFCC involving type-C equations centered at both faces and corners. Bottom: Comparison with a diagram for CAFCC from \cite{Kels:2020zjn} that involves type-A and type-B equations on faces, and a type-C equation on corners.  The directed edges for this diagram have only been included for the purpose of comparison, and are unnecessary as it involves only one type of type-B equation and one type of type-C equation.} 
\label{fig:CAFCCcube}
\end{figure}

\subsubsection{CAFCC algorithm}


The six components of the parameters 
$\al=(\ccpa,\ccpb),\bt=(\ccqa,\ccqb),\gm=(\ccra,\ccrb)$, 
are fixed, while 
%
\begin{align}\label{CAFCCinitial}
\ccyb,\ccy,\ccd,\ccb,\ccyc,\ccya,
\end{align}
are chosen as initial variables.  There remain a total of eight undetermined variables on vertices of the face-centered cube, and each of the fourteen equations \eqref{6face} and \eqref{8corner} to be satisfied.  

For the initial conditions \eqref{CAFCCinitial}, the six steps below may be used to determine whether the equations \eqref{typeCCAFCC}--\eqref{typeACAFCC} satisfy CAFCC.

\begin{enumerate}

\item
The two equations
\begin{align}
\begin{split}
\A{\ccy}{\ccya}{\ccf}{\ccyb}{\ccyc}{(\ccqa,\ccrb)}{(\ccpb,\ccra)}=0,\\
\C{\ccyb}{\ccy}{\ccd}{\ccb}{\ccc}{\ccpp}{\ccqq}=0,
\end{split}
\end{align}
are used to uniquely determine the variables $\ccf$ and $\ccc$, respectively.

\item
The three equations
\begin{align}
\begin{split}
\C{\ccya}{\ccf}{\ccy}{\cca}{\ccb}{\ccpp}{\ccrr}=0, \\
\A{\ccyc}{\ccf}{\cce}{\ccy}{\ccd}{\ccrr}{\ccqq}=0, \\
\A{\ccf}{\ccya}{\ccy}{\cczb}{\ccyc}{(\ccqa,\ccra)}{(\ccpb,\ccrb)}=0,
\end{split}
\end{align}
are used to uniquely determine the three variables $\cca$, $\cce$, and $\cczb$, respectively.

\item 
The two equations
\begin{align}
\begin{split}
\A{\ccd}{\ccza}{\cce}{\ccyb}{\ccyc}{(\ccqb,\ccrb)}{(\ccpb,\ccra)}=0, \\
\A{\cce}{\ccza}{\ccd}{\cczb}{\ccyc}{(\ccqb,\ccra)}{(\ccpb,\ccrb)}=0,
\end{split}
\end{align}
must agree for the solution of the variable $\ccza$.

\item 
The two equations
\begin{align}
\begin{split}
\Cf{\cca}{\ccb}{\cczc}{\ccya}{\cczb}{(\ccpa,\ccrb)}{(\ccra,\ccqa)}=0,\\
\Cf{\ccb}{\cczc}{\cca}{\ccyb}{\ccya}{(\ccpa,\ccra)}{(\ccqa,\ccrb)}=0,
\end{split}
\end{align}
must agree for the solution of the variable $\cczc$.

\item 
The four equations
\begin{align}
\begin{split}
\C{\ccza}{\cce}{\ccd}{\ccz}{\ccc}{\ccpp}{\ccrr}=0, \\
\Af{\cczc}{\cca}{\ccz}{\ccb}{\ccc}{\ccrr}{\ccqq}=0, \\
\C{\cczb}{\ccf}{\cce}{\cca}{\ccz}{\ccpp}{\ccqq}=0, \\
\Cf{\ccc}{\ccz}{\cczc}{\ccza}{\ccyb}{(\ccpa,\ccra)}{(\ccrb,\ccqb)}=0,
\end{split}
\end{align}
must agree for the solution of the final variable $\ccz$.

\item 
The remaining equation
\begin{align}\label{CAFCCalgorithmlast}
\begin{split}
\Cf{\ccz}{\ccc}{\cczc}{\ccza}{\cczb}{(\ccpa,\ccrb)}{(\ccra,\ccqb)}=0,
\end{split}
\end{align}
must be satisfied by the variables obtained in the previous steps.

\end{enumerate}

If following steps (i) and (ii) the remaining steps (iii)--(vi) can be satisfied, then the equations \eqref{typeCCAFCC}--\eqref{typeACAFCC} satisfy the property of CAFCC.

As for the original formulation of CAFCC \cite{Kels:2020zjn}, the above property of CAFCC may be derived using solutions of a particular form of an interaction-round-a-face (IRF) type classical Yang-Baxter equation. Further details on this connection are given in Appendix \ref{sec:YBE}.


\subsection{Integrable type-C lattice, CAFCC, and consistency across adjacent cubes}\label{sec:multicube}

Because type-C equations possess the least symmetry in comparison to type-A and type-B equations, only particular arrangements of type-C equations in the square lattice were found previously to possess vanishing algebraic entropy as required for integrability.  At present a general set of conditions for when an arrangement of equations will have vanishing entropy are not known, but Figure \ref{fig:lattice-Casym} shows one example of an arrangement that was found \cite{GubbiottiKels} to give vanishing algebraic entropy for a pair of type-C equations as indicated in Figure \ref{fig:2typeC}.

\begin{figure}[h!]
\centering
\includegraphics{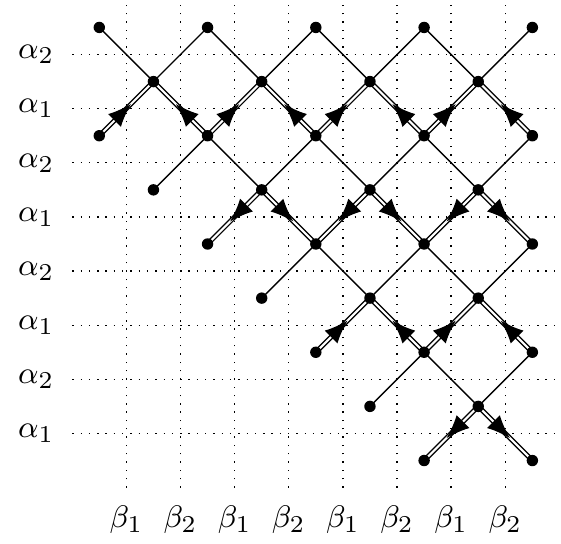}
\caption{A lattice arrangement that is integrable (in the sense of algebraic entropy) for the two different type-C equations of Figure \ref{fig:2typeC}. According to Figure \ref{fig:2typeC}, 
the equations $\Cf{x}{x_a}{x_b}{x_c}{x_d}{\ccpp}{\ccqq}=0$, $\C{x}{x_d}{x_c}{x_b}{x_a}{\ccpp}{\ccqq}=0$, $\C{x}{x_a}{x_b}{x_c}{x_d}{\ccpp}{\ccqq}=0$, and $\Cf{x}{x_d}{x_c}{x_b}{x_a}{\ccpp}{\ccqq}=0$, are respectively centered at the second, third, fourth, and fifth rows of vertices from the top.  This pattern repeats for every group of four rows of vertices in the lattice. }
\label{fig:lattice-Casym}
\end{figure}



Specific pairs of type-C equations that have vanishing algebraic entropy in the arrangement of Figure \ref{fig:lattice-Casym} are listed in the last two columns of Table \ref{table-AC} (where the expressions for the equations are given in Appendix \ref{app:equations}).  These equations satisfy CAFCC in combination with type-A equations as follows.

\begin{thm}\label{thm:CAFCCproperty}
Let the four equations $C$ in \eqref{6face} and the four equations $\overline{C}$ in \eqref{8corner} be chosen from one of the pairs of equations $C$ and $\overline{C}$ that are listed in the rows of Table \ref{table-AC}.  Then the choices of the equations $A$ and $\overline{A}$ in \eqref{6face} and \eqref{8corner} can be uniquely fixed by requiring their respective leg functions $a(x;y;\alpha,\beta)$ from \eqref{4leg} to match the leg functions of $C$ and $\overline{C}$ from \eqref{4legc}, as indicated in Table \ref{table-BC2} of Appendix \ref{app:equations}.  Each of the resulting ten combinations of equations listed in the rows of Table \ref{table-AC} satisfy the property of CAFCC.
\end{thm}


\begin{table}[h!]
\centering
\includegraphics{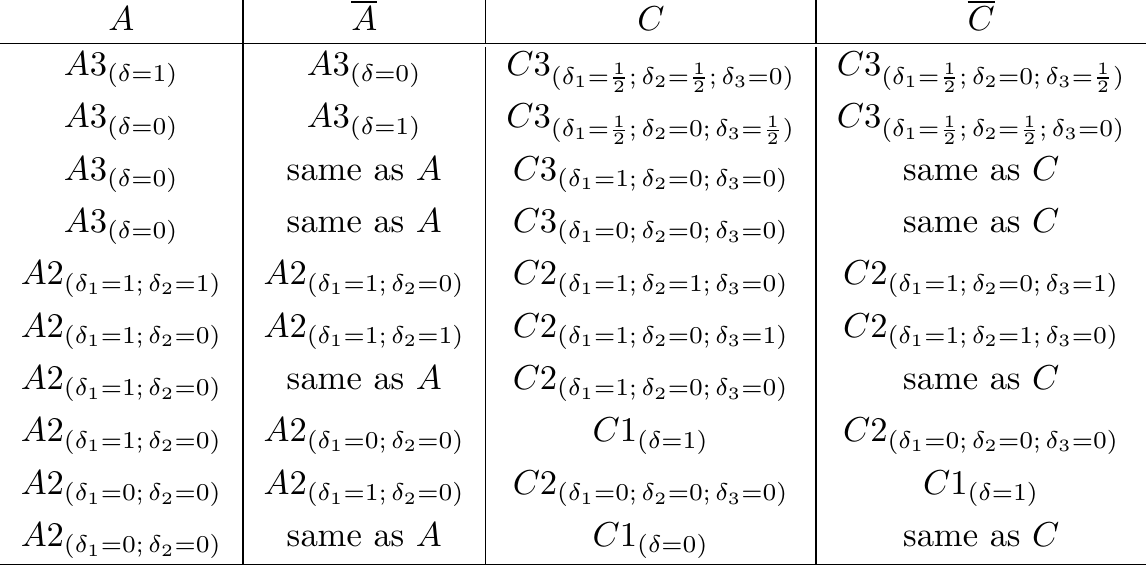}
\caption{Combinations of type-A and type-C equations (defined in Appendix \ref{app:equations}) for \eqref{6face} and \eqref{8corner} which satisfy CAFCC.  Pairs of type-C equations in columns $C$ and $\overline{C}$ have vanishing algebraic entropy in the arrangement of Figure \ref{fig:lattice-Casym}.  ``Same as $A$'' (or ``same as $C$'') indicates that the same equation that is listed for $A$ ($C$) should be used for $\overline{A}$ ($\overline{C}$).  }
\label{table-AC}
\end{table}


Theorem \ref{thm:CAFCCproperty} may be checked by direct computation.  Theorem \ref{thm:CAFCCproperty} gives the desired multidimensional consistency property for type-C equations that appear on faces of the face-centered cube, which allows them to be considered as integrable systems in higher-dimensional lattices and may be used to derive their Lax pairs, which will be shown in the following section.  

One of the consequences of the property of CAFCC is that because the evolution of face-centered quad equations involves equations that overlap, when considering consistency on adjacent face-centered cubes in higher-dimensional lattices there must be more instances of consistent equations than those just defined on the individual cubes.  This also applies (but was not considered) to the original formulation of CAFCC \cite{Kels:2020zjn}.  

This can be illustrated for the simplest case of extending the equations into a third dimension, with the example of the four CAFCC cubes on the left of Figure \ref{fig:4cubes}.  These four individual cubes will be denoted as $cF^{(nw)}$, $cF^{(ne)}$, $cF^{(sw)}$, $cF^{(se)}$, for the north-west, north-east, south-west, and south-east cubes respectively.  The types of equations on the pair of cubes $(cF^{(sw)},cF^{(se)})$, and on the pair of cubes $(cF^{(nw)},cF^{(ne)})$, are related by the exchanges $A\leftrightarrow \overline{A}$ and $C\leftrightarrow\overline{C}$.  Since each cube involves fourteen equations, there are a total of 56 equations on the individual cubes (but four pairs of these equations that appear on coinciding faces are equivalent), and these equations involve 38 different variables (associated to distinct vertices).

\begin{figure}[h!]
\centering
\includegraphics{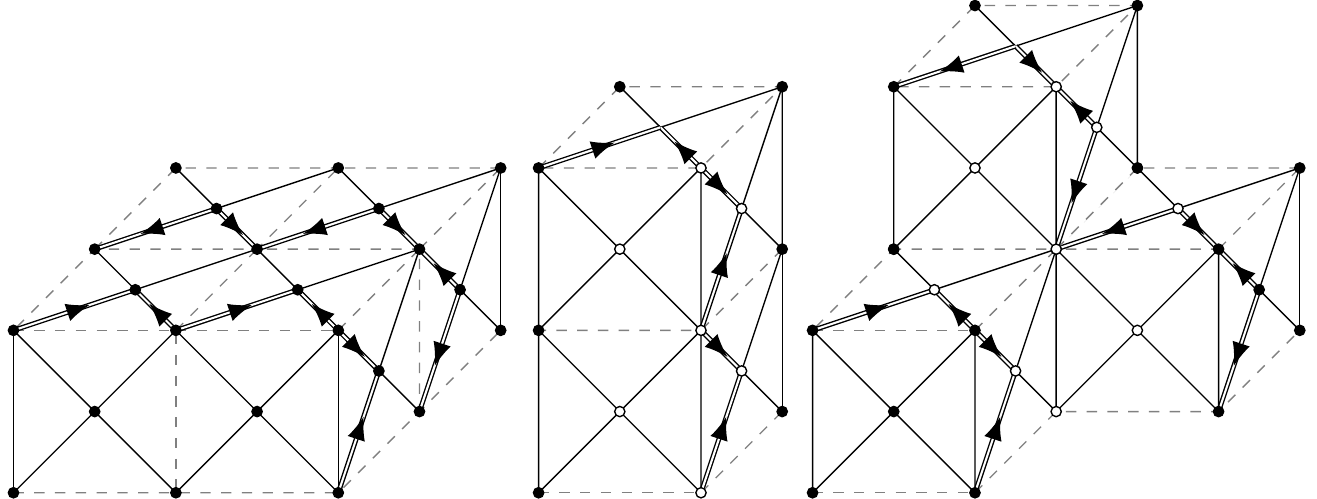}
\caption{Left: Four CAFCC cubes. The north-west, north-east, south-west, and south-east cubes are respectively denoted by $cF^{(nw)}$, $cF^{(ne)}$, $cF^{(sw)}$, and $cF^{(se)}$.  Center: A seven-point equation on white vertices of $cF^{(sw)}$ and another CAFCC cube. Right: A nine-point equation on white vertices of $cF^{(ne)}$, $cF^{(sw)}$, and another CAFCC cube.}
\label{fig:4cubes}
\end{figure}

The variables on each of the four cubes will be labelled according to the top diagram of Figure \ref{fig:CAFCCcube} and distinguished by the superscripts $(nw),(ne),(sw),(se)$ for the respective cubes. An initial value problem for the 56 equations can be posed by choosing ten initial variables as
\begin{align}\label{4cubeICs}
\ccyb^{(nw)},\ccy^{(nw)},\ccd^{(nw)}(=\ccy^{(ne)}),\ccb^{(nw)}(=\ccy^{(sw)}),\ccyc^{(nw)},\ccya^{(nw)},
\ccyb^{(ne)},\ccd^{(ne)}, 
\ccyb^{(sw)},\ccb^{(sw)}.
\end{align}

\begin{prop}
With the assignment of variables and equations on the four cubes according to \eqref{6face}, \eqref{8corner}, and the top diagram of Figure \ref{fig:CAFCCcube}, and with the initial conditions \eqref{4cubeICs}, there is an overdetermined system of 56 equations for 28 unknown variables on the four cubes $cF^{(nw)}$, $cF^{(ne)}$, $cF^{(sw)}$, $cF^{(se)}$.  If the four equations \eqref{typeCCAFCC} and \eqref{typeACAFCC} satisfy CAFCC then this overdetermined system of equations is consistent.
\end{prop}

Besides the 56 equations on the individual cubes, there are additional equations that will be satisfied which involve variables on two or more different cubes. However, these equations are not independent as they are satisfied as a consequence of the equations that are satisfied on the individual cubes.  As an example, there are the following four equations
\begin{align}\label{4cubeexample1}
\begin{split}
\A{y_b^{(nw)}}{z_n^{(nw)}}{z_n^{(ne)}}{y^{(nw)}}{y^{(ne)}}{(\alpha_2,\gamma_2)}{\hat{\bt}}=0,\\
\A{x_b^{(nw)}}{z_n^{(nw)}}{z_n^{(ne)}}{x^{(nw)}}{x^{(ne)}}{(\alpha_2,\gamma_1)}{\hat{\bt}}=0, \\
\Cf{y_d^{(nw)}}{z_s^{(ne)}}{z_s^{(nw)}}{y^{(ne)}}{y^{(nw)}}{(\alpha_1,\gamma_2)}{\bt}=0,\\
\Cf{x_d^{(nw)}}{b_s^{(nw)}}{z_s^{(ne)}}{x^{(nw)}}{x^{(ne)}}{(\alpha_1,\gamma_1)}{\hat{\bt}}=0,
\end{split}
\end{align}
respectively centered at the four common variables $y_b^{(nw)}=y_a^{(ne)}$, $x_b^{(nw)}=x_a^{(ne)}$, $y_d^{(nw)}=y_c^{(ne)}$, $x_d^{(nw)}=x_c^{(ne)}$, of the pair of cubes $(cF^{(nw)},cF^{(ne)})$.  There are similar sets of four equations on three other pairs of cubes $(cF^{(nw)},cF^{(sw)})$, $(cF^{(ne)},cF^{(se)})$, and $(cF^{(sw)},cF^{(se)})$, where each pair shares a face.  As mentioned above, such equations are satisfied as a consequence of the CAFCC equations on the individual cubes.  For example, the third equation of \eqref{4cubeexample1} is a consequence of the following two type-C equations 
\begin{align}\label{4cubeeq4}
\begin{split} 
\Cf{\ccz^{(nw)}}{\ccc^{(nw)}}{\cczc^{(nw)}}{\ccza^{(nw)}}{\cczb^{(nw)}}{(\ccpa,\ccrb)}{(\ccra,\ccqb)}=0, \\
\Cf{\cca^{(ne)}}{\ccb^{(ne)}}{\cczc^{(ne)}}{\ccya^{(ne)}}{\cczb^{(ne)}}{(\ccpa,\ccrb)}{(\ccra,\ccqa)}=0,
\end{split}
\end{align}
which are centered at the common vertex $y_d^{(nw)}=y_c^{(ne)}$ of the pair of cubes $(cF^{(nw)},cF^{(ne)})$.  This may be seen by rewriting the above two equations into their four-leg forms \eqref{4legc} and then dividing the second equation by the first.  Similarly, the other equations in \eqref{4cubeexample1} are consequences of two equations centered at the other three common vertices of the pair of cubes $(cF^{(nw)},cF^{(ne)})$.

As another example, there are the following two equations
\begin{align}\label{4cubeeq1}
\begin{split}
\Cf{\ccz^{(nw)}}{\cczb^{(se)}}{\cczb^{(sw)}}{\cczb^{(ne)}}{\cczb^{(nw)}}{\al}{\bt}=0, \\
\Cf{x_d^{(nw)}}{z_b^{(se)}}{z_b^{(sw)}}{z_b^{(ne)}}{z_b^{(nw)}}{\al}{\bt}=0,
\end{split}
\end{align}
which involve distinct variables of all four cubes on the left of Figure \ref{fig:4cubes}.  Once again these are satisfied as a consequence of CAFCC equations which are centered at a common vertex of the four cubes.  For example, the first equation in \eqref{4cubeeq1} is a consequence of the two type-C equations \eqref{4cubeeq4}, and two type-A equations on $cF^{(sw)}$ and $cF^{(se)}$ which are each centered at the same vertex $y_d^{(nw)}=y_c^{(ne)}=y_b^{(sw)}=y_a^{(se)}$.

This section considered only examples of CAFCC equations on multiple cubes that involve five different vertices, however, by using different arrangements of cubes there can also appear equations that involve from six to nine vertices, ({\it e.g.}, see the examples of Figure \ref{fig:4cubes}).  It is not clear what the equations involving six to nine vertices would be (although for systems of only type-A four-leg equations these should be related to Adler's discrete Toda equations on planar graphs \cite{AdlerPlanarGraphs}), but such equations could also turn out to possess nice properties associated to discrete integrability, which warrants further investigation. 

\section{Lax matrices}\label{sec:Lax}

Using the original formulation of CAFCC \cite{Kels:2020zjn}, it has been shown how to derive the Lax matrices for type-A and type-B CAFCC equations \cite{KelsLax}.  This method relied on the fact that CAFCC was formulated with type-A and type-B equations centered at faces of the face-centered cube.  Importantly, the new formulation of CAFCC given in \S\ref{sec:FCQECAFCC} has type-C equations centered at faces of the face-centered cube, which will allow for the derivation of their Lax matrices.  Lax matrices derived from both type-A and type-C CAFCC equations from \cite{KelsLax} will be required to construct the Lax pairs of type-C equations, and  
for convenience the method used to obtain these Lax matrices is included below.

\subsubsection{Type-A Lax matrix}

Due to the symmetries \eqref{refsym}--\eqref{symA}, for the type-A equation it does not matter which two variables are chosen for constructing the matrix equation.  Here the two variables $x_c$ and $x_d$ will be chosen to be written as
\begin{align}\label{eqtolaxsubs}
    x_c=\frac{f}{g}, \qquad x_d=\frac{f_L}{g_L}.
\end{align}
Then because of linearity in both $x_c$ and $x_d$, the equation $\A{x}{x_a}{x_b}{x_c}{x_d}{\al}{\bt}=0$ may be solved uniquely for the variable $x_d=\frac{f_L}{g_L}$ in the following form 
\begin{align}\label{xdsol}
    x_d=\frac{f_L}{g_L}
       =\frac{L_{11}(x;x_a,x_b;\al,\bt)f+L_{12}(x;x_a,x_b;\al,\bt)g}
             {L_{21}(x;x_a,x_b;\al,\bt)f+L_{22}(x;x_a,x_b;\al,\bt)g}.
\end{align}
The above equation can then be rewritten in a matrix form
\begin{align}\label{mateq1}
    \psi_{L}=\LA(x;x_a,x_b;\al,\bt)\psi,
\end{align}
where
\begin{align}
    \psi=\left(\!\!\begin{array}{c} f \\ g \end{array}\!\!\right)\!,\qquad 
    \psi_L=\left(\!\!\begin{array}{c} f_L \\ g_L \end{array}\!\!\right)\!,
\end{align}
and $\LA(x;x_a,x_b;\al,\bt)$ is the $2\times2$ matrix  
\begin{align}\label{Lentries}
\begin{split}
\LA(x;x_a,x_b;\al,\bt) =D_{\LA}\!\left(\!\begin{array}{cc}
  \ds L_{11}(x;x_a,x_b;\al,\bt) &
  \ds L_{12}(x;x_a,x_b;\al,\bt) \\
  \ds L_{21}(x;x_a,x_b;\al,\bt) &
  \ds L_{22}(x;x_a,x_b;\al,\bt)
 \end{array}\!\right)\!.
 \end{split}
\end{align}
The $D_{\LA}$ is an as yet unspecified normalisation factor that arises from separating the numerator and denominator of \eqref{xdsol}.  The matrix $\LA(x;x_a,x_b;\al,\bt)$ is one of the desired Lax matrices that will be used for the Lax pair of type-C equations.  Since it is obtained from a type-A equation it will be referred to as a type-A Lax matrix.  The reinterpretation of the type-A equation on the left of Figure \ref{fig:3fig4quad} as a matrix equation for the vectors $\Psi$ and $\Psi_L$ is shown in Figure \ref{laxfig1}.  Note that due to the symmetry \eqref{refsym}, the matrix $\LA(x;x_b,x_a;\al,\hat{\bt})$ is proportional to the inverse of $\LA(x;x_a,x_b;\al,\bt)$. 

\vspace{-0.3cm}

\begin{figure}[h!]
\centering
\includegraphics{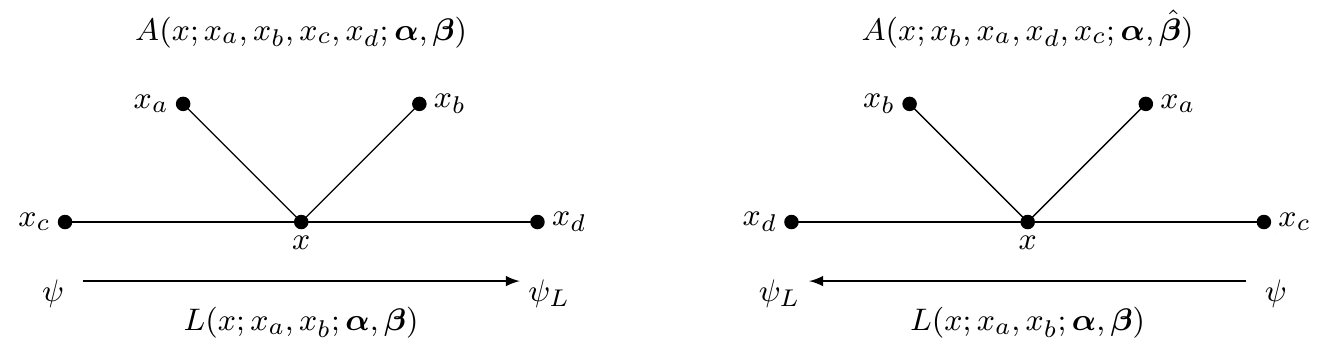}
\caption{The type-A equation from Figure \ref{fig:3fig4quad} is reinterpreted as equation \eqref{mateq1} for the matrix \eqref{Lentries}.  Both of these diagrams are equivalent, due to the symmetry \eqref{refsym}.}  
\label{laxfig1}
\end{figure}

\subsubsection{Type-C Lax matrix}

Since they have less symmetry, the expression for the Lax matrix derived from type-C equations will depend on which choices of variables are used to construct the matrix equation.  Let the variables $x_a$ and $x_c$ be written as
\begin{align}\label{eqtolaxsubs2}
    x_a=\frac{f_L}{g_L}, \qquad x_c=\frac{f}{g}.
\end{align}
Then because of linearity in both $x_a$ and $x_c$, the equation $\C{x}{x_a}{x_b}{x_c}{x_d}{\al}{\bt}=0$ may be solved uniquely for the variable $x_a=\frac{f_L}{g_L}$ in the following form 
\begin{align}\label{xdsol2}
    x_a=\frac{f_L}{g_L}
       =\frac{L_{11}(x;x_b,x_d;\al,\bt)f+L_{12}(x;x_b,x_d;\al,\bt)g}
             {L_{21}(x;x_b,x_d;\al,\bt)f+L_{22}(x;x_b,x_d;\al,\bt)g}.
\end{align}
The above equation can then be rewritten in a matrix form
\begin{align}\label{mateq2}
    \psi_{L}=\LC(x;x_b,x_d;\al,\bt)\psi,
\end{align}
where
\begin{align}
    \psi=\left(\!\!\begin{array}{c} f \\ g \end{array}\!\!\right)\!,\qquad 
    \psi_L=\left(\!\!\begin{array}{c} f_L \\ g_L \end{array}\!\!\right)\!,
\end{align}
and the $\LC(x;x_b,x_d;\al,\bt)$ is a $2\times2$ matrix given by 
\begin{align}\label{Lentries2}
\begin{split}
\LC(x;x_b,x_d;\al,\bt) =D_{\LC}\!\left(\!\begin{array}{cc}
  \ds L_{11}(x;x_b,x_d;\al,\bt) &
  \ds L_{12}(x;x_b,x_d;\al,\bt) \\
  \ds L_{21}(x;x_b,x_d;\al,\bt) &
  \ds L_{22}(x;x_b,x_d;\al,\bt)
 \end{array}\!\right)\!.
 \end{split}
\end{align}
The $D_{\LC}$ is an as yet unspecified normalisation factor that arises from separating the numerator and denominator of \eqref{xdsol2}.  The matrix $\LC(x;x_b,x_d;\al,\bt)$ is the second Lax matrix that will be used for the Lax pair of type-C equations.  Since it is obtained from a type-C equation it will be referred to as a type-C Lax matrix (although it is not the only Lax matrix that can be obtained from type-C equations \cite{KelsLax}).  The reinterpretation of the type-C equation on the right of Figure \ref{fig:3fig4quad} as a matrix equation for the vectors $\Psi$ and $\Psi_L$ is shown in Figure \ref{laxfig4}.

\vspace{-0.2cm}

\begin{figure}[h!]
\centering
\includegraphics{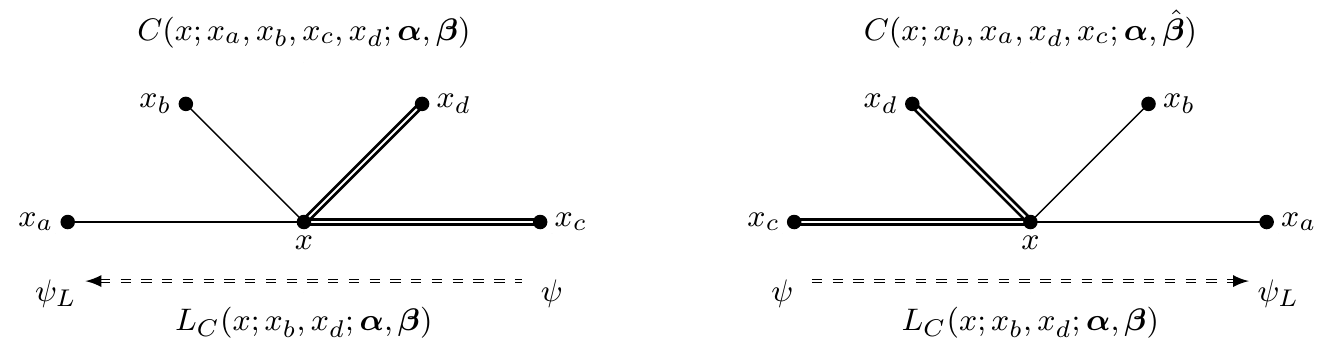}

\caption{The type-C equation from Figure \ref{fig:3fig4quad} is reinterpreted as equation \eqref{mateq2} for the matrix \eqref{Lentries2}.   Both of these diagrams are equivalent, due to the symmetry \eqref{refsym}.} 
\label{laxfig4}
\end{figure}

\vspace{-0.2cm}

If the above procedure is used for the inverse transition associated to $x_a\to x_c$ (or $\psi_L\to\psi$), this will give a matrix proportional to the inverse of \eqref{Lentries2}.  This inverse Lax matrix will also be used in the Lax equation, and thus for convenience it will be denoted separately by $\LCi(x;x_a,x_c;\al,\bt)$, where
\begin{align}\label{Lentries3}
\LCi(x;x_a,x_c;\al,\bt)=\LC(x;x_a,x_c;\al,\bt)^{-1}.
\end{align}

\subsection{Lax compatibility equation for type-C CAFCC equations}

The compatibility equation for the Lax matrices \eqref{Lentries} and \eqref{Lentries2} that reproduces a type-C equation can be derived from the property of CAFCC of \S\ref{sec:FCQECAFCC}.  The main idea from \cite{KelsLax} is to take two different evolutions between two vertices at the centers of opposite faces of the face-centered cube.  In terms of the top diagram of Figure \ref{fig:CAFCCcube}, the two type-A equations centered at $\ccy$ and $\ccf$ will be used for one evolution $\ccyb\to\ccyc\to\cczb$, and the two type-C equations centered at $\ccb$ and $\cca$ will be used for the other evolution $\ccyb\to\cczc\to\cczb$.  The consistency of these two different evolutions then provides a compatibility equation for the type-C equation centered at $\ccya$.

Specifically, the following four equations from \eqref{8corner}
\begin{align}\label{laxiniteqs}
\begin{array}{rr}
    \A{\ccy}{\ccya}{\ccf}{\ccyb}{\ccyc}{(\ccqa,\ccrb)}{(\ccpb,\ccra)}=0,\; 
    &\Cf{\ccb}{\cczc}{\cca}{\ccyb}{\ccya}{(\ccpa,\ccra)}{(\ccqa,\ccrb)}=0, \\[0.1cm]
    \A{\ccf}{\ccya}{\ccy}{\cczb}{\ccyc}{(\ccqa,\ccra)}{(\ccpb,\ccrb)}=0, \;
    &\Cf{\cca}{\ccb}{\cczc}{\ccya}{\cczb}{(\ccpa,\ccrb)}{(\ccra,\ccqa)}=0,
\end{array}
\end{align}
are used to define matrices for the transitions $\ccyb\to\ccyc$, $\ccyb\to\cczc$, $\ccyc\to\cczb$, $\cczc\to\cczb$, respectively.  These matrices are given in terms of \eqref{Lentries}, \eqref{Lentries2}, and \eqref{Lentries3}, as
\begin{align}\label{laxmatAC}
\begin{array}{rrr}
&\LL=\LA(\ccy;\ccya,\ccf;(\ccqa,\ccrb),(\ccpb,\ccra)), \qquad
&\M=\LC(\ccb;\cca,\ccya;(\ccpa,\ccra),(\ccqa,\ccrb)), \\[0.1cm]
&\Mb=\LA(\ccf;\ccy,\ccya;(\ccqa,\ccra),(\ccrb,\ccpb)), \qquad
&\Lb=\LCi(\cca;\ccb,\ccya;(\ccpa,\ccrb),(\ccqa,\ccra)).
\end{array}
\end{align}
The arrangement of equations \eqref{laxiniteqs} and associated matrices \eqref{laxmatAC} is shown schematically in Figure \ref{fig:LaxfigC}.  Note that according to Figure \ref{fig:3fig4quad}, the equation that is centered at $\ccya$ as a result of this arrangement is a type-C equation.

\vspace{-0.2cm}

\begin{figure}[h!]
\centering
\includegraphics{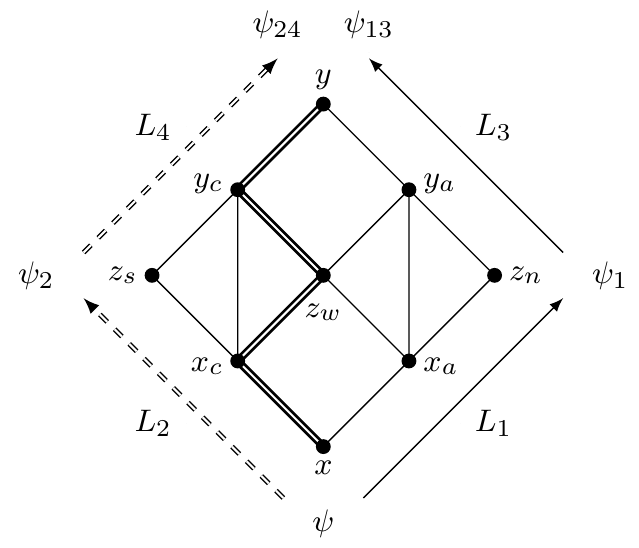}
\caption{Type-A and type-C equations \eqref{laxiniteqs} centered at $\ccy,\ccf,\ccb,\cca$ on the CAFCC cube of top diagram of Figure \ref{fig:CAFCCcube}, are reinterpreted as four matrices $\LL,\Mb,\M,\Lb$ in \eqref{laxmatAC} using Figures \ref{laxfig1} and \ref{laxfig4}.  The compatibility condition  $\psi_{13}\doteq\psi_{24}$ implies $\Lb\M-\Mb\LL\doteq 0$, where $\doteq$ indicates equality on solutions of the equation \eqref{laxeq2} centered at $\ccya$, which according to Figure \ref{fig:3fig4quad} is type-C.}
\label{fig:LaxfigC}
\end{figure}

\vspace{-0.2cm}

Next denoting
\begin{align}
    \psi_{13}=\Mb\psi_1=\Mb\LL\psi,\qquad \psi_{24}=\Lb\psi_2=\Lb\M\psi,
\end{align}
the compatibility condition is
\begin{align}\label{cc}
    \psi_{13}\doteq\psi_{24}.
\end{align}
The $\doteq$ denotes equality on solutions of the equation centered at $\ccya$ in Figure \ref{fig:LaxfigC}, which is
\begin{align}\label{laxeq2}
    \C{\ccya}{\ccf}{\ccy}{\cca}{\ccb}{\ccpp}{\ccrr}=0.
\end{align}
This corresponds to an equation in \eqref{6face} that appears on the face of the CAFCC cube in the top diagram of Figure \ref{fig:CAFCCcube}.  The compatibility condition \eqref{cc} is equivalently written in terms of Lax matrices \eqref{laxmatAC} as
\begin{align}\label{laxcomp2}
\Lb\M-\Mb\LL\doteq 0.
\end{align}
This is the desired equation which reinterprets the CAFCC property of the equations \eqref{laxiniteqs}, as a compatibility condition for the Lax matrices \eqref{laxmatAC} satisfied on the solution of a type-C CAFCC equation.  Analogously to CAC equations and their Lax pairs, $\ccqa$ is identified as the spectral parameter that comes from the additional third parameter associated to the CAFCC cube of the top diagram of Figure \ref{fig:CAFCCcube}.  A comparison of the Lax pairs arising from CAC and CAFCC is shown in Figure \ref{fig:CAFCCcomparison}.

\begin{figure}[h!]
\centering
\includegraphics{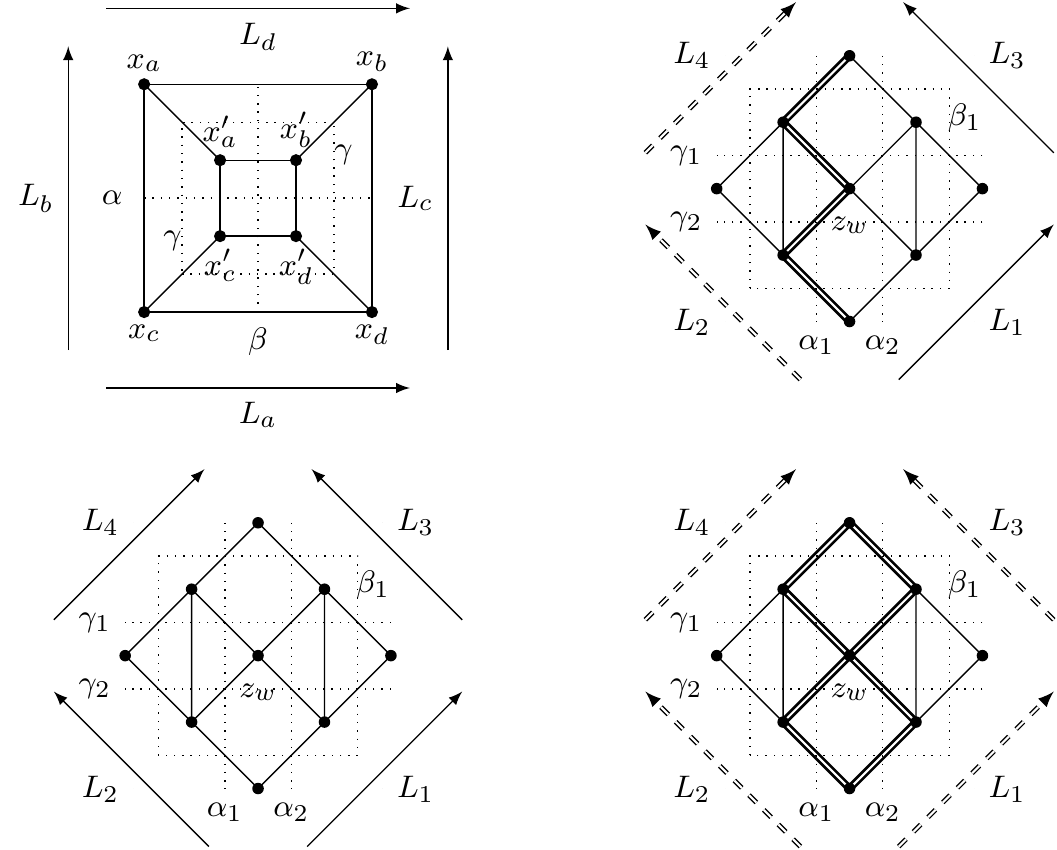}
\caption{Lax matrix equations for regular quad equations satisfying CAC (top-left), and face-centered quad equations satisfying CAFCC (top-right and bottom). On the top-left, the Lax compatibility equation would be satisfied on solutions of a quad equation $Q(x_a',x_b',x_c',x_d';\alpha,\beta)=0$, and $\gamma$ would be identified as the spectral parameter.  On the top-right, the Lax compatibility equation \eqref{laxcomp2} is satisfied on solutions of the type-C equation \eqref{laxeq2} (using the variable labelling of Figure \ref{fig:LaxfigC}), and $\beta_1$ is identified as the spectral parameter.  On the bottom, the Lax equations are satisfied on solutions of type-A and -B equations \cite{KelsLax}.}
\label{fig:CAFCCcomparison}
\end{figure}


\subsubsection{Compatibility equation for \texorpdfstring{$C3_{(\delta_1;\,\delta_3;\,\delta_2)}$}{C3(delta1;delta2;delta3)}.}

The type-A equation $A3_{(\delta)}$ is given in \eqref{a3d}.  The Lax matrix obtained from this equation may be written as \cite[Sec.~4.1.1]{KelsLax}
\begin{align}\label{LentriesA}
\LA(x;x_a,x_b;\al,\bt)=\frac{1}{D_{\LA}}\sum_{i=0}^2(L_{x^i}+\delta\Delta_{x^i})x^i,
\end{align}
where
\begin{gather}\label{a3mata}
    \Lxa=M(x_a,x_b),\qquad \Lxc=x_ax_bPM(\tfrac{1}{x_a},\tfrac{1}{x_b})P,
\\
    \Lxb=4\left(\!\!\begin{array}{cc}
    \bigl(\tfrac{\alpha_1}{\alpha_2}-\tfrac{\alpha_2}{\alpha_1}\bigr)x_a + \bigl(\tfrac{\alpha_1 \alpha_2}{\beta_1\beta_2} - \tfrac{\beta_1 \beta_2}{\alpha_1\alpha_2}\bigr)x_b & \bigl(\tfrac{\beta_1}{\beta_2}-\tfrac{\beta_2}{\beta_1})x_a x_b \\[0.15cm] \tfrac{\beta_1}{\beta_2}-\tfrac{\beta_2}{\beta_1} & \bigl(\tfrac{\alpha_1 \alpha_2}{\beta_1\beta_2}-\tfrac{\beta_1 \beta_2}{\alpha_1\alpha_2}\bigr)x_a + \bigl(\tfrac{\alpha_1}{\alpha_2}-\tfrac{\alpha_2}{\alpha_1}\bigr)x_b
    \end{array}\!\!\right)\!,
\\
    \Dxa=0,\qquad\Dxb=\left(\!\!\begin{array}{cc}
    0 & (\tfrac{\alpha_2}{\alpha_1}-\tfrac{\alpha_1}{\alpha_2})(\tfrac{\beta_1}{\beta_2}-\tfrac{\beta_2}{\beta_1})(\tfrac{\alpha_1\alpha_2}{\beta_1\beta_2} - \tfrac{\beta_1\beta_2}{\alpha_1\alpha_2}) \\ 0 & 0
    \end{array}\!\!\right)\!,
\\
\label{a3matb}
    \Dxc=\tfrac{(\alpha_1^2-\beta_1^2)(\alpha_1^2-\beta_2^2)(\alpha_2^2-\beta_1^2)(\alpha_2^2-\beta_2^2)}{(\alpha_1\alpha_2\beta_1\beta_2)^2}
    \left(\!\!\begin{array}{cc}
    \tfrac{\alpha_1\beta_1}{\beta_1^2-\alpha_1^2}  &  \tfrac{x_a\alpha_2\beta_1}{\alpha_2^2-\beta_1^2}+\tfrac{x_b\alpha_2\beta_2}{\beta_2^2-\alpha_2^2} \\[0.15cm] 
    0 & \tfrac{\alpha_1\beta_2}{\beta_2^2-\alpha_1^2}
    \end{array}\!\!\right)\!,
\end{gather}
and
\begin{align}
P=\left(\!\!\begin{array}{cc}0 & 1 \\ 1 & 0\end{array}\!\!\right),\qquad
M(x_a,x_b)=4\left(\!\!\begin{array}{cc}
    \tfrac{\beta_1}{\alpha_1}-\tfrac{\alpha_1}{\beta_1} & \bigl(\tfrac{\alpha_2}{\beta_1}-\tfrac{\beta_1}{\alpha_2}\bigr) x_a + \bigl(\tfrac{\beta_2}{\alpha_2}-\tfrac{\alpha_2}{\beta_2}\bigr) x_b \\[0.15cm] 0 & \tfrac{\beta_2}{\alpha_1}-\tfrac{\alpha_1}{\beta_2}
    \end{array}\!\!\right).
\end{align}

The type-C equation $C3_{(\delta_1;\,\delta_2;\,\delta_3)}$ is given in \eqref{c3ddd}.  The Lax matrix obtained from this equation may be written as \cite[Sec.~4.2.1]{KelsLax}
\begin{align}\label{LentriesC}
\LC(x;x_b,x_d;\al,\bt)=\frac{1}{D_{\LC}}\sum_{i=0}^2(L_{x^i}+\delta_1\Delta_{x^i})x^i,
\end{align}
where
\begin{gather}
\label{c3mata}
    \Lxa= \left(\!\!\begin{array}{cc}
    -\beta_2 & \beta_1 x_d \\ 0 & 0 
    \end{array}\!\!\right)\!,
\quad
    \Lxb=-\left(\!\!\begin{array}{cc}
    -\alpha_2 x_b & \frac{\beta_1 \beta_2}{\alpha_2} x_b x_d \\[0.1cm] \frac{\beta_1 \beta_2}{\alpha_2} & -\alpha_2 x_d
    \end{array}\!\!\right)\!,
\quad
    \Lxc= x_b\left(\!\!\begin{array}{cc}
    0 & 0 \\[0.1cm] \beta_1 & -\beta_2 x_d
    \end{array}\!\!\right)\!,
\\
    \Dxa=\frac{2 {\delta_3}\beta_1\beta_2}{\alpha_1}\left(\!\!\begin{array}{cc}
    0 & (\tfrac{\beta_2}{\alpha_2} - \tfrac{\alpha_2}{\beta_2})x_b \\[0.1cm] 0 & \tfrac{\beta_1}{\alpha_2} - \tfrac{\alpha_2}{\beta_1}
    \end{array}\!\!\right)\!,
\quad
    \Dxb=\frac{2\beta_1\beta_2}{\alpha_1}(\tfrac{\beta_2}{\beta_1}-\tfrac{\beta_1}{\beta_2})\left(\!\!\begin{array}{cc}
     {\delta_2} x_d & \tfrac{\alpha_1^2}{2\beta_1\beta_2} \\ 0 &  {\delta_3}x_b
    \end{array}\!\!\right)\!,
\\
\label{c3matb}
    \Dxc=\left(\!\!\begin{array}{cc}
    {\delta_2}(\tfrac{\beta_1}{\alpha_2}-\tfrac{\alpha_2}{\beta_1})(\tfrac{\alpha_2}{\beta_2} - \tfrac{\beta_2}{\alpha_2} + 2\tfrac{ \beta_2}{\alpha_1} x_b x_d)\beta_1 & 
    (\tfrac{\beta_1}{\alpha_2}-\tfrac{\alpha_2}{\beta_1})\bigl(\alpha_1 x_b - {\delta_2}(\tfrac{\alpha_2}{\beta_2} - \tfrac{\beta_2}{\alpha_2})\beta_2x_d\bigr) \\
    2 {\delta_2}(\tfrac{\beta_2}{\alpha_2} - \tfrac{\alpha_2}{\beta_2})\tfrac{\beta_1\beta_2}{\alpha_1}x_d & \alpha_1(\tfrac{\beta_2}{\alpha_2}-\tfrac{\alpha_2}{\beta_2})
    \end{array}\!\!\right)\!.
\end{gather}

\begin{prop}\label{prop:c3ddd}

For the case $\delta=0$ of the Lax matrix \eqref{LentriesA} defined with \eqref{a3mata}--\eqref{a3matb}, and the three cases $(\delta_1,\delta_2,\delta_3)=(0,0,0),(1,0,0),(\frac{1}{2},\frac{1}{2},0)$, of the Lax matrix \eqref{LentriesC} defined with \eqref{c3mata}--\eqref{c3matb}, and the normalisations chosen as
\begin{align}\label{normc3}
\begin{split}
    D_{\LA}&=4\ii (\alpha_1^2 - \beta_1^2)(\beta_1x_a-\alpha_2x)(\beta_2x-\alpha_2x_b)(\alpha_1\alpha_2\beta_1\beta_2)^{-1}, \\
    D_{\LC}&=(\alpha_2)^{-1},
\end{split}
\end{align}
the Lax equation \eqref{laxcomp2} is satisfied on solutions of $C3_{(\delta_1;\,\delta_3;\,\delta_2)}(\ccya;\ccf,\ccy,\cca,\ccb;\al,\gm)=0$.
\end{prop}

\begin{proof}

Using the definitions \eqref{laxmatAC}, the Lax equation of Proposition \ref{prop:c3ddd} for the three cases $(\delta_1,\delta_2,\delta_3)=(0,0,0),(1,0,0),(\frac{1}{2},\frac{1}{2},0)$, may be written as
\begin{align}
\Lb\M-\Mb\LL=\frac{C3_{(\delta_1;\,\delta_3;\,\delta_2)}(\ccya;\ccf,\ccy,\cca,\ccb;\al,\gm)
\left(\!\!\begin{array}{c} 
    \gamma_1 \\ \frac{\beta_1}{\ccya}
    \end{array}\!\!\right)\otimes
    \left(\!\!\begin{array}{c} 
    -\gamma_2 \\ \beta_1\ccya 
    \end{array}\!\!\right)}
 {(\beta_1^2-\gamma_2^2)(\gamma_2 \ccy - \alpha_2 \ccya)(\alpha_2 \ccf - \gamma_1 \ccya)(\cca - \delta_1 \frac{\alpha_1}{\gamma_1\ccya} - \delta_2\frac{\gamma_1\ccya}{\alpha_1})}.
\end{align}

\end{proof}

\subsubsection{Compatibility equation for \texorpdfstring{$C2_{(\delta_1;\,\delta_3;\,\delta_2)}$}{C2(delta1;delta2;delta3)} and \texorpdfstring{$C1_{(1)}$}{C1(1)}}

In the following, $\theta_{ij}$ are defined as differences of parameters as follows
\begin{align}\label{thtdef1}
    \theta_{ij}=\theta_i-\theta_j, \qquad i,j\in\{1,2,3,4\},\qquad (\theta_1,\theta_2,\theta_3,\theta_4)=(\alpha_1,\alpha_2,\beta_1,\beta_2).
\end{align}

The type-A equation $A2_{(\hat{\delta}_1;\,\hat{\delta}_2)}$ is given in \eqref{a2dd}.  The parameters have been denoted as $\hat{\delta}_1$ and $\hat{\delta}_2$ to avoid confusion with the $\delta_i$ parameters for the type-C equation (which may take different values in the construction of the Lax equation).  The Lax matrix obtained from this equation is \eqref{LentriesA} with \cite[Sec.~4.1.2]{KelsLax}
\begin{gather}\label{a2mata}
    \Lxa=\left(\!\!\begin{array}{cc}
    \theta_{31} & \theta_{23}x_a -\theta_{24}x_b \\ 0 & \theta_{41}
    \end{array}\!\!\right)\!,
\quad
    \Lxb=\left(\!\!\begin{array}{cc}
    x_a \theta_{12} + x_b(\theta_{13}+\theta_{24}) & \theta_{34} x_a x_b \\ \theta_{34} & x_b \theta_{12} + x_a(\theta_{13}+\theta_{24})
    \end{array}\!\!\right)\!,
\\
    \Lxc=x_a x_b\left(\!\!\begin{array}{cc}
     \theta_{41} & 0 \\ \tfrac{1}{x_a} \theta_{23} - \tfrac{1}{x_b} \theta_{24} & \theta_{31}
    \end{array}\!\!\right)\!,
\qquad 
    \Dxa=\left(\!\!\begin{array}{cc}
    0 & {\hat{\delta}_2} \theta_{12} \theta_{34}(\theta_{13} + \theta_{24}) \\ 0 & 0
    \end{array}\!\!\right)\!,
\\
\Dxb=\left(\!\!\begin{array}{cc}
{\hat{\delta}_2}\bigl(2 \theta_{13}\theta_{24}^2 + \theta_{12}\theta_{34}(\theta_{12}-\theta_{34})\bigr) & (\Dxb)_{12} \\
0 & {\hat{\delta}_2}\bigl(2 \theta_{14}\theta_{23}^2-\theta_{12}\theta_{34}(\theta_{12}+\theta_{34})\bigr)
\end{array}\!\!\right)\!,
\\
\label{a2matb}
\Dxc=\left(\!\!\begin{array}{cc}
\theta_{14} \theta_{23} \theta_{24}(x_b - \theta_{12} \theta_{34} - \theta_{24}^2)^{\hat{\delta}_2} - {\hat{\delta}_2} \theta_{14}(x_a \theta_{12} \theta_{24} - x_b \theta_{23} \theta_{13} ) & (\Dxc)_{12} \\
{\hat{\delta}_2} \theta_{23} \theta_{24} \theta_{43} & (\Dxc)_{22}
\end{array}\!\!\right)\!,
\end{gather}
where
\begin{align}
\begin{split}
(\Dxb)_{12}=&\,\theta_{12} \theta_{34}(\theta_{13} + \theta_{24})(x_a+x_b-\theta_{13}^2-\theta_{24}^2-\theta_{12}\theta_{34})^{\hat{\delta}_2} -2{\hat{\delta}_2}\theta_{13} \theta_{14}(x_a \theta_{24} + x_b \theta_{32}), \\
(\Dxc)_{12}=&\, \theta_{13} \theta_{14} \Bigl(x_b \theta_{23}(\theta_{12} \theta_{43}- \theta_{13}^2)^{\hat{\delta}_2} - x_a \theta_{24}(\theta_{12} \theta_{34} - \theta_{14}^2)^{\hat{\delta}_2} \\
    &\phantom{\theta_{13} \theta_{14} \Bigl(} - \hat{\delta}_2 \theta_{34} \bigl(x_a x_b-(\theta_{13}+\theta_{24})\theta_{12}\theta_{23}\theta_{24}\bigr)\!\Bigr), \\ 
(\Dxc)_{22}=&\, \theta_{13} \theta_{23} \theta_{24}(x_a + \theta_{12} \theta_{34} - \theta_{23}^2)^{\hat{\delta}_2} + {\hat{\delta}_2} \theta_{13}\bigl(x_a \theta_{14} \theta_{24} - x_b \theta_{12} \theta_{23}\bigr).
\end{split}
\end{align}

The type-C equation $C2_{(\delta_1;\,\delta_2;\,\delta_3)}$ is given in \eqref{c2ddd}.  The Lax matrix obtained from this equation is \eqref{LentriesC} with  \cite[Sec.~4.2.2]{KelsLax}
\begin{gather}\label{c2mata}
    \Lxa=\left(\!\!\begin{array}{cc}
    1 & -\theta_{34}-x_d \\ 0 & 0 
    \end{array}\!\!\right)\!,
\;\;
    \Lxb=-\left(\!\!\begin{array}{cc}
    x_b & x_b (\theta_{23} + \theta_{24} - x_d) \\ -1 & \theta_{23} + \theta_{24} + x_d
    \end{array}\!\!\right)\!,
\;\;
    \Lxc=x_b\left(\!\!\begin{array}{cc}
    0 & 0 \\ -1 & x_d - \theta_{34}
    \end{array}\!\!\right)\!,
\\
    \Dxa=\left(\!\!\begin{array}{cc}
    -2({\delta_2}+{\delta_3}) & 2 {\delta_2}(x_d+\theta_{34}) + {\delta_3}\bigl(2(x_d - x_b \theta_{24}) +\theta_{34}(\theta_{13}+\theta_{14}+1)\bigr) \\ 0 & -2 {\delta_3} \theta_{23}
    \end{array}\!\!\right)\!,
\\
\Dxb=\left(\!\!\begin{array}{cc}
-\theta_{34}\bigl(2(x_d-\theta_{13})\bigr)^{\delta_2}(-1)^{\delta_3} + ({\delta_2}+{\delta_3})2 x_b + {\delta_2}(\theta_{23}^2+\theta_{24}^2) & (\Dxb)_{12} \\
-2({\delta_2}+{\delta_3}) & (\Dxb)_{22}
\end{array}\!\!\right)\!,
\\
\label{c2matb}
\Dxc=\left(\!\!\begin{array}{cc}
(\Dxc)_{11} & (\Dxc)_{12} \\
\theta_{42}(2 x_d -\theta_{12}-\theta_{14})^{\delta_2}(-1)^{\delta_3} + ({\delta_2}+{\delta_3})2 x_b & (\Dxc)_{22}
\end{array}\!\!\right)\!,
\end{gather}
where
\begin{align}
\begin{split}
(\Dxb)_{12}=&\, (\theta_{43} + {\delta_3} x_b)\bigl(x_d(\theta_{31}+\theta_{41})^{\delta_2}(-1)^{\delta_3} +2 {\delta_3} \theta_{12}^2 + (\theta_{31}^{1+{\delta_2}}+\theta_{41}^{1+{\delta_2}}) (\theta_{32}+\theta_{42})^{\delta_3}\bigr) \\
  &+ 2{\delta_2}\bigl(x_b(\theta_{23} - x_d) + \theta_{24}(x_b - x_d \theta_{23}) \bigr) -  {\delta_3}x_b\bigl(x_d + 2 \theta_{12}^2 - (\theta_{23} + \theta_{24})\bigr) , \\
(\Dxb)_{22}=&\, 2({\delta_2}+{\delta_3})x_d - 2{\delta_3} x_b \theta_{34} + (2 {\delta_2}+{\delta_3})(\theta_{23}+\theta_{24})(\theta_{13}+\theta_{14}+1)^{\delta_3} \\
(\Dxc)_{11}=&\, \theta_{32}\Bigl(\theta_{24}\bigl(\theta_{34}(\theta_{12}+\theta_{14}-2 x_d) +\theta_{24}^2 -x_b\bigr)^{\delta_2} + x_b(2 x_d -\theta_{13}-\theta_{14})^{\delta_2}\Bigr)(-1)^{\delta_3}, \\
(\Dxc)_{12}=&\, \theta_{32}\biggl(\theta_{24}\bigl(\theta_{34}( \theta_{23} \theta_{24} -2 \theta_{13} \theta_{14})^{\delta_2}(\theta_{31}+\theta_{41})^{\delta_3} - x_d(2 \theta_{13}\theta_{43} + \theta_{23}\theta_{24})^{\delta_2} (-1)^{\delta_3}\bigr) \\[-0.15cm]
 & + x_b\Bigl((\theta_{31}+\theta_{41})^{1+{\delta_2}+{\delta_3}} -{\delta_2}\theta_{24}\theta_{34} - ({\delta_2} +{\delta_3}) 2 \theta_{13}\theta_{14} + x_d(\theta_{21}+\theta_{31})^{\delta_2} (-1)^{\delta_3}\Bigr)\! \biggr), \\[-0.1cm]
(\Dxc)_{22}=&\, \theta_{24}\bigl(-(\theta_{41}^{1+{\delta_2}+{\delta_3}}+\theta_{31}^{1+{\delta_2}+{\delta_3}})  - x_d(2 \theta_{31}-\theta_{24})^{\delta_2}(-1)^{\delta_3} - {\delta_2} \theta_{23}\theta_{34} \bigr) \\
& - (2 {\delta_2} + {\delta_3})(x_d-\theta_{34})x_b - {\delta_3}\bigl(x_d - \theta_{34}(\theta_{13}+\theta_{14})\bigr)x_b.
\end{split}
\end{align}

\begin{prop}\label{prop:c2ddd}

For the $(\hat{\delta}_1,\hat{\delta}_2)=(1,0)$ case of the Lax matrix \eqref{LentriesA} defined with \eqref{a2mata}--\eqref{a2matb}, and the three cases $(\delta_1,\delta_2,\delta_3)=(0,0,0),(1,0,0),(1,1,0)$ of the Lax matrix \eqref{LentriesC} defined with \eqref{c2mata}--\eqref{c2matb}, and the normalisations chosen as
\begin{align}\label{normc2}
\begin{split}
    D_{\LA}&=\ii (\alpha_1 - \beta_1)(\beta_1+x_a-\alpha_2-x)(\beta_2+x-\alpha_2-x_b), \\
    D_{\LC}&=1,
\end{split}
\end{align}
the Lax equation \eqref{laxcomp2} is satisfied on solutions of $C2_{(\delta_1;\,\delta_3;\,\delta_2)}(\ccya;\ccf,\ccy,\cca,\ccb;\al,\gm)=0$ if $(\delta_1,\delta_2,\delta_3)=(1,0,0)\textrm{ or }(1,1,0)$, and on solutions of $C1_{(1)}(\ccya;\ccf,\ccy,\cca,\ccb;\al,\gm)=0$ if $(\delta_1,\delta_2,\delta_3)=(0,0,0)$.
\end{prop}

\begin{proof}

Using the definitions \eqref{laxmatAC}, the Lax equation of Proposition \ref{prop:c2ddd} for the three cases $(\delta_1,\delta_2,\delta_3)=(0,0,0),(1,0,0),(1,0,1)$, may be written as
\begin{align}
\Lb\M-\Mb\LL=\frac{\frac{\tilde{C}}{2(\beta_1-\gamma_2)}
\left(\!\!\begin{array}{c} 
    \beta_1-\gamma_1-\ccya \\ -1
    \end{array}\!\!\right)\otimes
    \left(\!\!\begin{array}{c} 
    1 \\ \gamma_2-\beta_1-\ccya  
    \end{array}\!\!\right)}
 {(\gamma_2 + \ccy - \alpha_2 - \ccya)(\alpha_2 + \ccf - \gamma_1 - \ccya)(\cca - \delta_1 (\alpha_1-\gamma_1-\ccya)^{1+\delta_2}\bigr)},
\end{align}
where 
\begin{align}
\tilde{C}=
    \left\{\begin{array}{ll}
    C2_{(\delta_1;\,\delta_3;\,\delta_2)}(\ccya;\ccf,\ccy,\cca,\ccb;\al,\gm), & (\delta_1,\delta_2,\delta_3)=(1,0,0)\textrm{ or }(1,1,0), \\[0.1cm]
    -C1_{(1)}(\ccya;\ccf,\ccy,\cca,\ccb;\al,\gm), & (\delta_1,\delta_2,\delta_3)=(0,0,0).
    \end{array}\right.
\end{align}

\end{proof}

\subsubsection{Compatibility equation for \texorpdfstring{$C2_{(0;\,0;\,0)}$}{C2(0;0;0)} and \texorpdfstring{$C1_{(0)}$}{C1(0)}}

The Lax matrix obtained from the type-C equation $C1_{(\delta)}$ given in \eqref{c1d} is (this matrix is slightly different from\cite[Sec.~4.2.3]{KelsLax} due to a slightly different expression for $C1_{(\delta=0)}$ given in \eqref{c1d}, while the $\delta=1$ case is new)
\begin{align}\label{c1mat}
\begin{split}
   \LC=\frac{(-1)^{\delta}}{D_{\LC}}\left(\!\!\begin{array}{cc} \bigl(x+\delta(\alpha_2-\beta_1) \bigr)\bigl(x - x_b+\delta(\beta_2-\alpha_2)\bigr) & L_{12} \\ 
   x - x_b+\delta(\beta_2-\alpha_2) & 2(\alpha_2-\beta_2)(-\frac{xd}{2})^\delta + (x_b - x) x_d \end{array}\!\!\right)\!,
\end{split}
\end{align}
where
\begin{align}
L_{12}=2\bigl((\beta_1-\beta_2)x +(\alpha_2-\beta_1)x_b\bigr)(-\tfrac{x_d}{2})^{\delta} + \bigl(\delta(\alpha_2-\beta_1)(\alpha_2-\beta_2) - x (x - x_b)\bigr) x_d.
\end{align}

\begin{prop}\label{prop:c1d}

For the $(\hat{\delta}_1,\hat{\delta}_2)=(0,0)$ case of the Lax matrix \eqref{LentriesA} defined with \eqref{a2mata}--\eqref{a2matb}, and the Lax matrix \eqref{c1mat}, with the normalisations chosen as
\begin{align}\label{normc1}
\begin{split}
    D_{\LA}&=(\alpha_1 - \beta_1)(x-x_a)(x-x_b), \\
    D_{\LC}&=x-x_b,
\end{split}
\end{align}
the Lax equation \eqref{laxcomp2} is satisfied on solutions of $C2_{(0;\,0;\,0)}(\ccya;\ccf,\ccy,\cca,\ccb;\al,\gm)=0$ if $\delta=1$ in \eqref{c1mat}, and on solutions of $C1_{(0)}(\ccya;\ccf,\ccy,\cca,\ccb;\al,\gm)=0$ if $\delta=0$ in \eqref{c1mat}.
\end{prop}

\begin{proof}

Using the definitions \eqref{laxmatAC}, the Lax equation of Proposition \ref{prop:c1d} 
may be written as
\begin{align}
\Lb\M-\Mb\LL=\frac{\tilde{C}}
 {2(\beta_1-\gamma_2)(\ccy - \ccya)(\ccf - \ccya)}
 \left(\!\!\begin{array}{c} 
    \ccya \\ 1
    \end{array}\!\!\right)\otimes
    \left(\!\!\begin{array}{c} 
    1 \\ -\ccya  
    \end{array}\!\!\right),
\end{align}
where
\begin{align}
\tilde{C}=
    \left\{\begin{array}{ll}
    C1_{(0)}(\ccya;\ccf,\ccy,\cca,\ccb;\al,\gm), & \delta=0, \\[0.1cm]
    \frac{1}{\ccya}C2_{(0;\,0;\,0)}(\ccya;\ccf,\ccy,\cca,\ccb;\al,\gm), & \delta=1.
    \end{array}\right.
\end{align}

\end{proof}

\section{Conclusion}

Consistency-around-a-face-centered-cube (CAFCC) has recently been introduced as a new form of multidimensional consistency for integrability of systems of five-point face-centered quad equations in the square lattice \cite{Kels:2020zjn}.  This paper introduces a new formulation of CAFCC which involves so-called type-C equations that are centered at faces of the face-centered cube, whereas previously they were only centered at corners.  This is essential for consistently extending systems of type-C equations into higher-dimensional lattices as well as for deriving their Lax pairs, as has been done in \S\ref{sec:Lax}.  The establishing of the Lax pairs and multidimensional consistency of type-C equations in this paper, along with the previous establishing of their vanishing algebraic entropy \cite{GubbiottiKels}, shows that systems of type-C CAFCC equations (as well as type-A and -B equations) satisfy the important properties that are associated with integrability of lattice equations.

An important problem will be to determine whether there are other equations that satisfy CAFCC, for which the results here (and \cite{KelsLax}) would imply their Lax pairs.  This might be done by establishing a classification result similar to \cite{ABS}, which has not yet been explored for multilinear face-centered quad equations of the form \eqref{afflin}.  It is also worth investigating potential integrability of the different equations that arise on adjacent cubes in higher-dimensional lattices, as illustrated in \S\ref{sec:FCQECAFCC}\ref{sec:multicube}.  Besides providing new examples of integrable systems, this could also potentially lead to other interesting forms of consistency connected to some different forms of the Yang-Baxter equation.  It would also be interesting if there is some relation between CAFCC and the recently proposed consistency-around-a-cuboctahedron property \cite{JoshiNakazonoCUBO}, which also involves some face-centered cubic structure, but applies to regular quad equations rather than face-centered quad equations. This could also inspire new reductions to discrete Painlev\'e equations \cite{JoshiNakazonoCUBOReduction}.

\begin{appendices}
\numberwithin{equation}{section}

\section{Type-A and type-C CAFCC equations}\label{app:equations}

Here the definition \eqref{thtdef1} is used, {\it i.e.}, $\theta_{ij}=\theta_i-\theta_j$, where $(\theta_1,\theta_2,\theta_3,\theta_4)=(\alpha_1,\alpha_2,\beta_1,\beta_2)$.

\begin{multline}\label{a3d}
A3_{(\delta)}=
\tfrac{\delta}{4}(\tfrac{\alpha_1}{\alpha_2}-\tfrac{\alpha_2}{\alpha_1})(\tfrac{\beta_1}{\beta_2}-\tfrac{\beta_2}{\beta_1})\bigl(\tfrac{\beta_1\beta_2}{\alpha_1\alpha_2}-\tfrac{\alpha_1\alpha_2}{\beta_1\beta_2}\bigr)x \\ 
\begin{split}
 +\Bigl((\tfrac{\beta_1}{\beta_2}-\tfrac{\beta_2}{\beta_1})(x_ax_b-x_cx_d) + (\tfrac{\alpha_1}{\alpha_2}-\tfrac{\alpha_2}{\alpha_1})(x_ax_c-x_bx_d) - (\tfrac{\alpha_1\alpha_2}{\beta_1\beta_2}-\tfrac{\beta_1\beta_2}{\alpha_1\alpha_2})(x_ax_d-x_bx_c)\Bigr)x  \\
+(\tfrac{\alpha_2}{\beta_1}-\tfrac{\beta_1}{\alpha_2})(x_a x^2 - x_bx_cx_d) - (\tfrac{\alpha_2}{\beta_2}-\tfrac{\beta_2}{\alpha_2})(x_b x^2 - x_ax_cx_d)   
 - (\tfrac{\alpha_1}{\beta_1}-\tfrac{\beta_1}{\alpha_1})(x_c x^2 - x_ax_bx_d)   \\
  + (\tfrac{\alpha_1}{\beta_2}-\tfrac{\beta_2}{\alpha_1})(x_d x^2 - x_ax_bx_c) 
 + \tfrac{\delta}{4}\Bigl((\tfrac{\alpha_1}{\beta_1}-\tfrac{\beta_1}{\alpha_1})(\tfrac{\alpha_2}{\beta_2}-\tfrac{\beta_2}{\alpha_2})\bigl((\tfrac{\alpha_1}{\beta_2}-\tfrac{\beta_2}{\alpha_1})x_a + (\tfrac{\alpha_2}{\beta_1}-\tfrac{\beta_1}{\alpha_2})x_d\bigr)
 \\
 \phantom{+\delta\Bigl(} - (\tfrac{\alpha_1}{\beta_2}-\tfrac{\beta_2}{\alpha_1})(\tfrac{\alpha_2}{\beta_1}-\tfrac{\beta_1}{\alpha_2})\bigl((\tfrac{\alpha_1}{\beta_1}-\tfrac{\beta_1}{\alpha_1})x_b + (\tfrac{\alpha_2}{\beta_2}-\tfrac{\beta_2}{\alpha_2})x_c\bigr) 
\Bigr) =0.
\end{split}
\end{multline}

\begin{multline}\label{a2dd}
 A2_{(\delta_1;\,\delta_2)}=\Bigl( (\theta_{13}+\theta_{24})(x_bx_c-x_ax_d)  + \theta_{12}(x_ax_c-x_bx_d) + \theta_{34}(x_ax_b-x_cx_d) \Bigr)x \\ 
 \begin{split}
 + \theta_{23}(x_ax^2 -x_bx_cx_d) - \theta_{24}(x_bx^2 -x_ax_cx_d) - \theta_{13}(x_cx^2 -x_ax_bx_d)   
+ \theta_{14}(x_dx^2 -x_ax_bx_c) \\ + {\delta_1}x\theta_{12}\theta_{34}(\theta_{13}+\theta_{24})\bigl(x+x_a+x_b+x_c+x_d - \theta_{12}^2-\theta_{13}\theta_{23}-\theta_{14}\theta_{24}\bigr)^{\delta_2} \\
 + {\delta_1}\Bigl(\theta_{14}\theta_{23}(\theta_{13}x_b+\theta_{24}x_c)(2x-\theta_{12}\theta_{34})^{\delta_2} - \theta_{13}\theta_{24}(\theta_{13}x_a+\theta_{23}x_d)(2x+\theta_{12}\theta_{34})^{\delta_2}
  \Bigr)  \\
 + {\delta_2}\Bigl( x_a \theta_{13}\theta_{14}(\theta_{24}\theta_{14}^2-\theta_{34}x_b) - x_b\theta_{13}\theta_{23}(\theta_{14}\theta_{13}^2-\theta_{12}x_d) - x_c\theta_{14}\theta_{24}(\theta_{23}\theta_{24}^2+\theta_{12}x_a)     \\
 \phantom{+} + x_d\theta_{23}\theta_{24}(\theta_{13}\theta_{23}^2 +\theta_{34}x_c) +
     \bigl(x_a x_d\theta_{13}\theta_{42} + x_b x_c\theta_{23}\theta_{14} + {\textstyle \prod_{1\leq i<j\leq4}}\theta_{ij}\bigr)(\theta_{13}+\theta_{24})\!\Bigr)\! =0.
\end{split}
\end{multline}

\begin{multline}\label{c3ddd}
C3_{(\delta_1;\,\delta_2;\,\delta_3)}=
\Bigl(\alpha_2 (\beta_1 x_d - \beta_2 x_c) -{\delta_3}\bigl(\alpha_2^2 (\beta_1 x_b-\beta_2 x_a) + \beta_1 \beta_2(\beta_1 x_a - \beta_2 x_b)\bigr)\alpha_1^{-1}\Bigr)x^2 \\ 
\begin{split}
 + \alpha_2 x_a x_b (\beta_2 x_d-\beta_1 x_c) + {\delta_1}\alpha_1\bigl(\beta_1 x_b-\beta_2 x_a + \alpha_2^2 (\tfrac{x_a}{\beta_2} - \tfrac{x_b}{\beta_1})\bigr)  +\Bigl(\alpha_2^2 (x_b x_c - x_a x_d)  \\
 + \beta_1 \beta_2 (x_a x_c - x_b x_d) + \alpha_2(\tfrac{\beta_2}{\beta_1}-\tfrac{\beta_1}{\beta_2})\bigl({\delta_1} \alpha_1 - {\delta_3}\tfrac{\beta_1 \beta_2}{\alpha_1}x_a x_b  +  {\delta_2}\tfrac{\beta_1 \beta_2}{\alpha_1} x_c x_d\bigr)\Bigr)x \\
 + {\delta_2}\Bigl(\tfrac{(\alpha_2^2 - \beta_1^2)(\alpha_2^2 - \beta_2^2)}{2\alpha_2\beta_1\beta_2}(\beta_2 x_d-\beta_1 x_c) + \tfrac{x_c x_d}{\alpha_1}\bigl(\beta_1 \beta_2 (\beta_1 x_b-\beta_2 x_a) + \alpha_2^2 (\beta_1 x_a - \beta_2 x_b)\bigr)\Bigr)=0.
    \end{split}
\end{multline}


\begin{multline}\label{c2ddd}
C2_{(\delta_1;\,\delta_2;\,\delta_3)}=
(x_d -x_c)(x^2 + x_a x_b) + \theta_{34} (x^2 - x_a x_b)(\theta_{13}+\theta_{14})^{\delta_3} +2 {\delta_3}(\theta_{23} x_a - \theta_{24} x_b)x^2 \\
    \begin{split}
 + \Bigl((x_a + x_b + 2 {\delta_2} \theta_{23}\theta_{24})(x_c - x_d) - (x_a-x_b)(\theta_{23}+\theta_{24})(\theta_{13}+\theta_{14})^{\delta_3} + 2 {\delta_3} \theta_{34} x_a x_b \Bigr)x \\
 + {\delta_1}
\Bigl(\theta_{13}^{1+{\delta_2}+{\delta_3}}+\theta_{14}^{1+{\delta_2}+{\delta_3}} + 2 {\delta_2} x_c x_d -(x_c+x_d)(\theta_{13}+\theta_{14})^{\delta_2}\Bigr)
\bigl( \theta_{34}(\delta_2\theta_{23}\theta_{24}-x)-x_b \theta_{23} \\
+x_a \theta_{24}\bigr) +\delta_1 \theta_{23}\theta_{24}\bigl(x_c-x_d+\theta_{34}(\theta_{13}+\theta_{14}- 2 x)^{{\delta_3}}\bigr)(x_a+x_b-\theta_{34}^2-\theta_{23}\theta_{24})^{\delta_2} =0.
    \end{split}
\end{multline}





 \begin{multline}\label{c1d}
 C1_{(\delta)}=
 (x_c-x_d)x^2 +
 \Bigl(2(\beta_1-\beta_2)\bigl(-\frac{x_c+x_d}{2}\bigr)^{\delta} -(x_a+x_b)(x_c-x_d)\Bigr)x + 2\bigl(-\frac{x_c+x_d}{2}\bigr)^{\delta}  \\
  \times\bigl((\beta_2-\alpha_2)x_a +(\alpha_2-\beta_1)x_b\bigr)+ \bigl(x_ax_b- \delta (\alpha_2-\beta_1)(\alpha_2-\beta_2)\bigr)(x_c-x_d) =0.
\end{multline}


\subsection*{Four-leg expressions}

In Table \ref{table-BC2}, the abbreviation {\it add.} indicates an additive form of one of the equations \eqref{4leg} or \eqref{4legc}, given respectively by
\begin{align}
    a(x;x_a;\alpha_2,\beta_1)+a(x;x_d;\alpha_1,\beta_2)-a(x;x_b;\alpha_2,\beta_2)-a(x;x_c;\alpha_1,\beta_1)=0, \\
    a(x;x_a;\alpha_2,\beta_1)+c(x;x_d;\alpha_1,\beta_2)-a(x;x_b;\alpha_2,\beta_2)-c(x;x_c;\alpha_1,\beta_1)=0.
\end{align}
%


 
 
 















\begin{table}[htb!]
\centering
\includegraphics{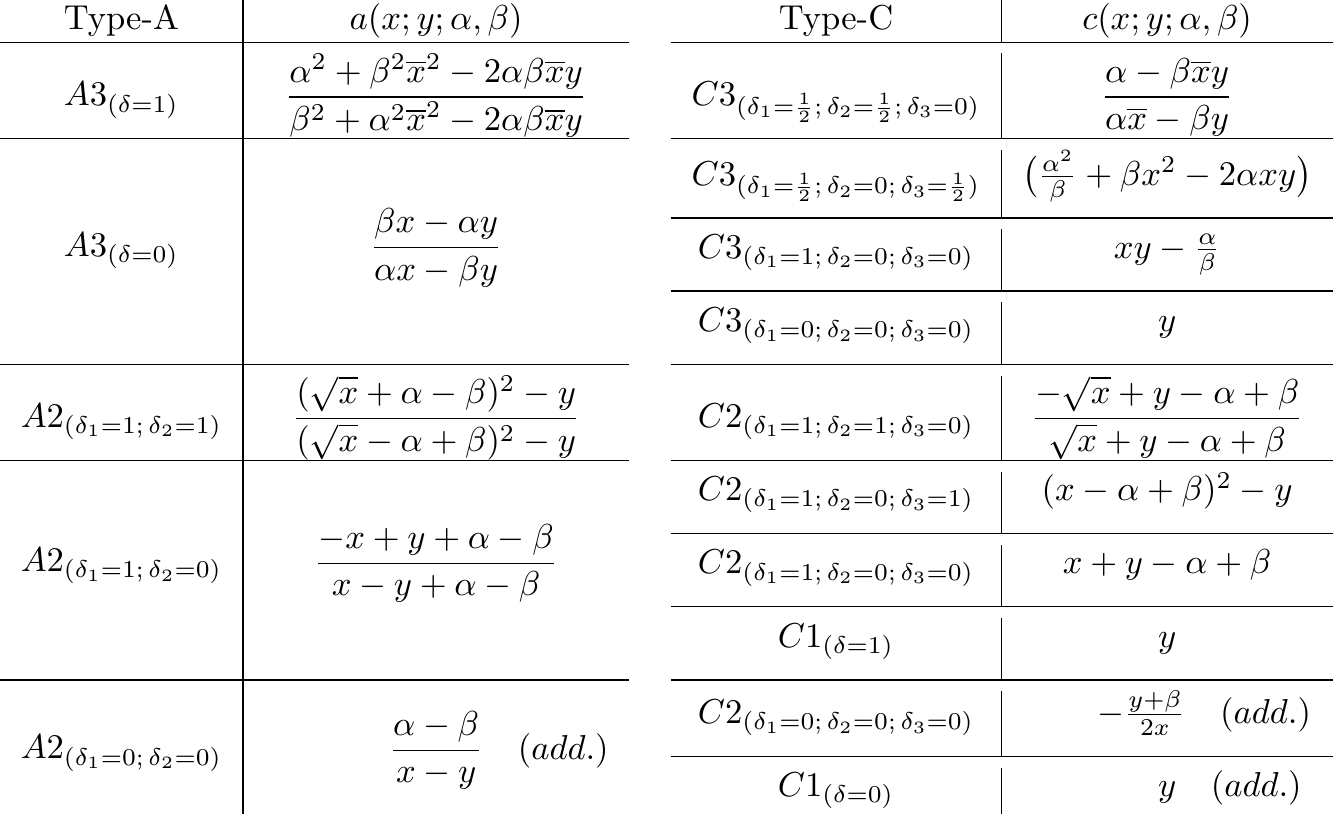}
\caption{Left: 
A list of $a(x;y;\alpha,\beta)$ in \eqref{4leg} for type-A equations \eqref{a3d} and \eqref{a2dd}. Right: A list of $c(x;y;\alpha,\beta)$ in \eqref{4legc} for type-C equations \eqref{c3ddd}, \eqref{c2ddd}, and \eqref{c1d}.  For $C3_{(\delta_1;\,\delta_2;\,\delta_3)}$, $C2_{(\delta_1;\,\delta_2;\,\delta_3)}$, and  $C1_{(\delta)}$, the $a(x;y;\alpha,\beta)$ are respectively given by $A3_{(2\delta_2)}$, $A2_{(\delta_1;\,\delta_2)}$, and $A2_{(\delta;\,0)}$.  Here $\overline{x}=x+\sqrt{x^2-1}$.}%
\label{table-BC2}
\end{table}

\section{Classical Yang-Baxter equation and CAFCC}\label{sec:YBE}

The original formulation of CAFCC was derived from new types of interaction-round-a-face (IRF) forms of the classical Yang-Baxter equation (CYBE), and explicit solutions of the latter were constructed from solutions of the classical star-triangle relations.  For the new formulation of CAFCC of \S\ref{sec:FCQECAFCC}\ref{sec:CAFCC}, there is a similar connection to a different form of the CYBE.

The relevant CYBE may be defined in terms of six complex-valued functions
\begin{align}\label{CSTRsol}
\lag_{\alpha}(x_i,x_j),\quad\lagh_{\alpha}(x_i,x_j),\quad\ol_{\alpha}(x_i,x_j),\quad\olh_{\alpha}(x_i,x_j),\quad\lam_{\alpha}(x_i,x_j),\quad\olam_{\alpha}(x_i,x_j),   
\end{align}
that each depend on two complex variables $x_i,x_j$ and a complex parameter $\alpha$.  These six functions are associated to edges and vertices shown in Figure \ref{fig:edges}, where the parameter $\alpha$ is represented by the difference of rapidity variables $u$ and $v$ associated to directed edges.  It is assumed that the partial derivatives also satisfy the following relations
\begin{align}\label{rels1}
\begin{split}
\frac{\partial D_{\alpha}(x_i,x_j)}{\partial x_i}=\frac{\partial D_{\alpha}(x_j,x_i)}{\partial x_i}(\textrm{mod } 2\pi\ii), \qquad
\frac{\partial D_{\alpha}(x_i,x_j)}{\partial x_i}=-\frac{\partial D_{-\alpha}(x_j,x_i)}{\partial x_i}(\textrm{mod } 2\pi\ii),
\end{split}
\end{align}
where $D$ represents one of $\lag,\ol,\lagh,\olh$ from \eqref{CSTRsol}, and
\begin{align}\label{rels2}
\begin{split}
\frac{\partial F_{\alpha}(x_i,x_j)}{\partial x_I}=\frac{\partial F_{2\eta+\alpha}(x_i,x_j)}{\partial x_I}(\textrm{mod } 2\pi\ii), \;
\frac{\partial G_{\alpha}(x_i,x_j)}{\partial x_I}=-\frac{\partial H_{\eta+\alpha}(x_i,x_j)}{\partial x_I}(\textrm{mod } 2\pi\ii),\;\; I\in\{i,j\},
\end{split}
\end{align}
where $\eta$ is some constant, $F$ is one of $\lag,\ol,\lagh,\olh,\lam,\olam$ from \eqref{CSTRsol}, and $(G,H)$ represents one of the three pairs $(\lag,\ol),(\lagh,\olh),(\lam,\olam)$.  The relations \eqref{rels1}--\eqref{rels2} reduce the number of different types of equations that need to be considered for CAFCC \cite{Kels:2020zjn}.

\begin{figure}[h!]
\centering
\includegraphics{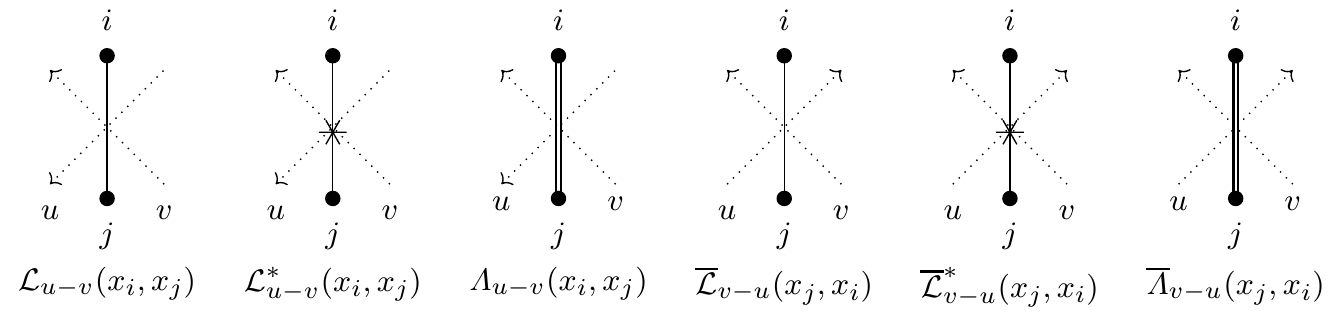}
\caption{Functions associated to edges for solutions of the classical Yang-Baxter equation.}
\label{fig:edges}
\end{figure}

Let $B_L$ and $B_R$ denote the graphs that appear on the left and right hand sides of the equality of Figure \ref{fig:YBE}, with respective sets of vertices $V(B_L)=\{a,b,c,d,e,f,h,i,j,k\}$ and $V(B_R)=\{a,b,c,d,e,f,l,m,n,o\}$, and sets of edges $E(B_L)$ and $E(B_R)$ which connect two vertices.  Then the CYBE is defined by
\begin{align}\label{YBEdef}
\sum_{(ij)\in E(B_L)}K_{\alpha_{ij}}(x_i,x_j)-\sum_{(ij)\in E(B_R)}K_{\alpha_{ij}}(x_i,x_j)=2\pi\ii\sum_{i\in V(B_L)\cup V(B_R)}k_ix_i+C(\alpha_{ij}),
\end{align}
for some integers $k_i\in\mathbb{Z}$, where $C(\alpha_{ij})$ is a constant with respect to the variables $x_i$ ($i\in V(B_L)\cup V(B_R)$), $K_{\alpha_{ij}}(x_i,x_j)$ is the function associated to an edge $(ij)$ according to the assignment of Figure \ref{fig:edges},  and the variables $x_i$ ($i\in V(B_L)\cup V(B_R)$) are subject to the following eight constraints
\begin{align}\label{YBEeqmo}
\begin{split}
\frac{\partial}{\partial x_J}\sum_{(ij)\in E(B_L)\cup E(B_R)}K_{\alpha_{ij}}(x_i,x_j)
=2\pi\ii k_J,\quad J\in\{h,i,j,k,l,m,n,o\}.
\end{split}
\end{align}

\begin{figure}[h!]
\centering
\includegraphics{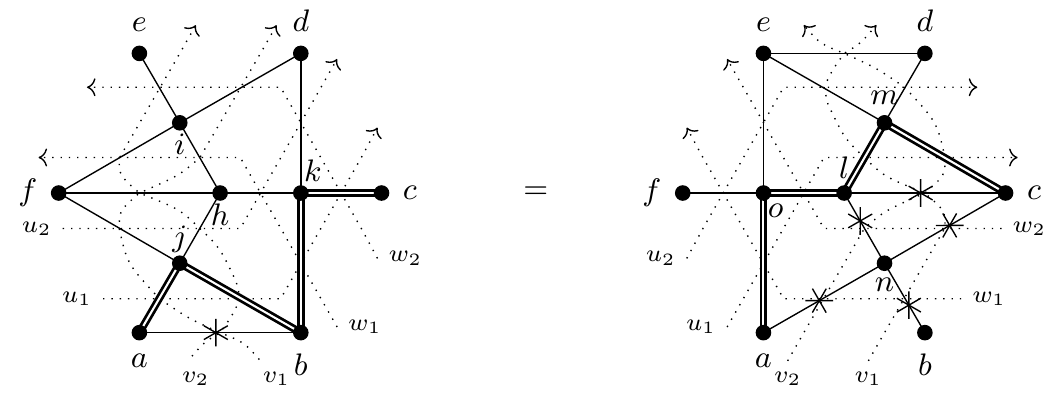}
\caption{Expression for the classical Yang-Baxter equation corresponding to \eqref{YBEdef}--\eqref{YBEeqmo}.}
\label{fig:YBE}
\end{figure}

The CAFCC property may be thought of as a reinterpretation of the classical IRF YBE of \eqref{YBEdef}.  First, there are a further six partial derivatives of the CYBE in addition to \eqref{YBEeqmo}, which are given by the equations
\begin{align}\label{YBEeqmo2}
\begin{split}
\frac{\partial}{\partial x_J}\Bigl(\sum_{(ij)\in E(B_L)}K_{\alpha_{ij}}(x_i,x_j)-\sum_{(ij)\in E(B_R)}K_{\alpha_{ij}}(x_i,x_j)\Bigr)
=2\pi\ii k_J,\quad J\in V(B_L)\cap V(B_R).
\end{split}
\end{align}

Introduce the following four functions defined in terms of \eqref{CSTRsol}
\begin{align}\label{43legs}
\begin{split}
 \At{x}{x_a}{x_b}{x_c}{x_d}{\al}{\bt}=\EXP^{\frac{\partial}{\partial y} \bigl(\lag_{v_1-u_2}(y,y_a)+\ol_{v_2-u_2}(y,y_b)+\ol_{v_1-u_1}(y_c,y)+\lag_{v_2-u_1}(y_d,y)\bigr)}, \\
\Aft{x}{x_a}{x_b}{x_c}{x_d}{\al}{\bt}=\EXP^{\frac{\partial}{\partial y} \bigl(\lag^{\ast}_{v_1-u_2}(y,y_a)+\ol^{\ast}_{v_2-u_2}(y,y_b)+\ol^{\ast}_{v_1-u_1}(y_c,y)+\lag^{\ast}_{v_2-u_1}(y_d,y)\bigr)}, \\
 \Ct{x}{x_a}{x_b}{x_c}{x_d}{\al}{\bt}=\EXP^{\frac{\partial}{\partial y} \bigl(\lag_{v_1-u_2}(y,y_a)+\ol_{v_2-u_2}(y,y_b)+\olam_{v_1-u_1}(y_c,y)+\lam_{v_2-u_1}(y_d,y)\bigr)}, \\
\Cft{x}{x_a}{x_b}{x_c}{x_d}{\al}{\bt}=\EXP^{\frac{\partial}{\partial y} \bigl(\lag^{\ast}_{v_1-u_2}(y_a,y)+\ol^{\ast}_{v_2-u_2}(y_b,y)+\olam_{v_1-u_1}(y,y_c)+\lam_{v_2-u_1}(y,y_d)\bigr)}.
\end{split}
\end{align}
The variables on the left and right hand sides of \eqref{43legs} for $a$ and $\overline{a}$ are related by
\begin{align}\label{cov1}
x=f(y),\quad x_i=f(y_i),\; i\in\{a,b,c,d\},
\end{align}
and the variables on the left and right hand sides for $c$ and $\overline{c}$ are related by
\begin{align}\label{cov2}
x=f(y),\quad 
x_i=\left\{\begin{array}{ll}
f(y_i), & i\in\{a,b\},\\
g(y_i), & i\in\{c,d\},
\end{array}\right.
\end{align}
where $f(y)$ and $g(y)$ are chosen so that each of $a,\overline{a},c,\overline{c}$ are ratios of multilinear polynomials in the four variables $x_a,x_b,x_c,x_d$.  The relation between parameters in \eqref{43legs} is
\begin{align}\label{cov3}
\alpha_1=h(u_1),\quad\alpha_2=h(u_2),\quad\beta_1=h(v_1),\quad\beta_2=h(v_2),
\end{align}
where $h(z)$ is chosen so that each of $a,\overline{a},c,\overline{c}$ have an algebraic dependence on $\alpha_1,\alpha_2,\beta_1,\beta_2$.  Because of the above choices of $f(y)$ and $g(y)$, the four equations
\begin{align}
\begin{gathered}
 \At{x}{x_a}{x_b}{x_c}{x_d}{\al}{\bt}=1,\qquad \Aft{x}{x_a}{x_b}{x_c}{x_d}{\al}{\bt}=1, \\
 \Ct{x}{x_a}{x_b}{x_c}{x_d}{\al}{\bt}=1,\qquad \Cft{x}{x_a}{x_b}{x_c}{x_d}{\al}{\bt}=1,
 \end{gathered}
\end{align}
can respectively be written in the equivalent forms
\begin{align}\label{4CAFCCfunctions}
\begin{gathered}
\A{x}{x_a}{x_b}{x_c}{x_d}{\al}{\bt}=0,\qquad \Af{x}{x_a}{x_b}{x_c}{x_d}{\al}{\bt}=0, \\
\C{x}{x_a}{x_b}{x_c}{x_d}{\al}{\bt}=0,\qquad \Cf{x}{x_a}{x_b}{x_c}{x_d}{\al}{\bt}=0,
\end{gathered}
\end{align}
where each of $A,\overline{A},C,\overline{C}$ are multilinear polynomials of the form \eqref{afflin}.  Then analogously to what was found for the CYBE for the original formulation of CAFCC \cite{Kels:2020zjn}, if the six functions \eqref{CSTRsol} are a solution of the CYBE \eqref{YBEdef}--\eqref{YBEeqmo}, the four equations \eqref{4CAFCCfunctions} will satisfy the property of CAFCC given in Section \ref{sec:CAFCC}.  This follows from the fact that each of the fourteen partial derivatives of the CYBE \eqref{YBEeqmo}--\eqref{YBEeqmo2} may be identified (mod $2\pi\ii$ and up to the change of variables \eqref{cov1}--\eqref{cov3}) with the fourteen equations \eqref{6face}--\eqref{8corner} of CAFCC, using the bijection between the vertices of Figure \ref{fig:YBE} and the top diagram of Figure \ref{fig:CAFCCcube}, given by
\begin{align}
(x,x_a,x_b,x_c,x_d)\mapsto(k,h,d,b,c), \;\;
(y,y_a,y_b,y_c,y_d)\mapsto(o,f,e,a,l), \;\;
(z_n,z_e,z_s,z_w)\mapsto(i,m,n,j).
\end{align}

\end{appendices}

\vskip6pt

\enlargethispage{20pt}


\dataccess{This article has no additional data.}


\competing{I declare I have no competing interests.}

\funding{No funding has been received for this article.}

\ack{The author would like to thank Giorgio Gubbiotti for fruitful discussions, and also the anonymous referees for several helpful suggestions that led to improvement of the manuscript.}




\begin{thebibliography}{9}

\bibitem{AdlerPlanarGraphs}
V.~E. Adler, ``Discrete equations on planar graphs,''
  \href{http://dx.doi.org/10.1088/0305-4470/34/48/310}{{\em J. Phys. A: Math.
  Gen.} {\bfseries 34} no.~48, (2001) 10453--10460}.

\bibitem{ABS}
V.~E. Adler, A.~I. Bobenko, and Y.~B. Suris, ``{Classification of Integrable
  Equations on Quad-Graphs. The Consistency Approach},''
  \href{http://dx.doi.org/10.1007/s00220-002-0762-8}{{\em Commun. Math. Phys.}
  {\bfseries 233} (2003) 513--543}.
  
\bibitem{BellonViallet1999}
M.~Bellon and C.-M. Viallet, ``Algebraic entropy,''
  \href{http://dx.doi.org/10.1007/s002200050652}{{\em Commun. Math. Phys.}
  {\bfseries 204} (1999) 425--437}.

\bibitem{BobSurQuadGraphs}
A.~I. Bobenko and Y.~B. Suris, ``{I}ntegrable systems on quad-graphs,''
  \href{http://dx.doi.org/10.1155/S1073792802110075}{{\em Int. Math. Res. Not.}
  {\bfseries 2002} no.~11, (2002) 573--611}.

\bibitem{BHQKLax}
T.~Bridgman, W.~Hereman, G.~R.~W. Quispel, and P.~H. van~der Kamp, ``Symbolic
  computation of {L}ax pairs of partial difference equations using consistency
  around the cube,'' \href{http://dx.doi.org/10.1007/s10208-012-9133-9}{{\em
  Found. Comput. Math.} {\bfseries 13} no.~4, (2013) 517--544}.

\bibitem{DoliwaSantini}
A.~Doliwa and P.~M. Santini, ``Multidimensional quadrilateral lattices are
  integrable,''
  \href{http://dx.doi.org/https://doi.org/10.1016/S0375-9601(97)00456-8}{{\em
  Phys. Lett. A} {\bfseries 233} no.~4, (1997) 365--372}.

\bibitem{GubbiottiKels}
G.~Gubbiotti and A.~P. Kels, ``Algebraic entropy for face-centered quad
  equations,'' \href{http://dx.doi.org/10.1088/1751-8121/ac2aeb}{{\em J. Phys.
  A: Math. Theor.} {\bfseries 54} (2021) 455201}.

\bibitem{hietarinta_joshi_nijhoff_2016}
J.~Hietarinta, N.~Joshi, and F.~W. Nijhoff,
  \href{http://dx.doi.org/10.1017/CBO9781107337411}{{\em Discrete Systems and
  Integrability}}.
\newblock Cambridge Texts in Applied Mathematics. Cambridge University Press,
  2016.

\bibitem{HietarintaNEWCAC}
J.~Hietarinta, ``Search for {CAC}-integrable homogeneous quadratic triplets of
  quad equations and their classification by {BT} and {L}ax,''
  \href{http://dx.doi.org/10.1080/14029251.2019.1613047}{{\em J. Nonl. Math.
  Phys.} {\bfseries 26} no.~3, (2019) 358--389}.

\bibitem{JoshiNakazonoCUBO}
N.~Joshi and N.~Nakazono, ``Classification of quad-equations on a
  cuboctahedron,'' \href{http://arxiv.org/abs/1906.06650}{{\ttfamily
  arXiv:1906.06650 [nlin.SI]}}.
  
\bibitem{JoshiNakazonoCUBOReduction}
N.~Joshi and N.~Nakazono, ``Reduction of quad-equations consistent around a
  cuboctahedron {I}: Additive case,''
  \href{http://dx.doi.org/10.1090/bproc/96}{{\em Proc. Amer. Math. Soc. Ser. B}
  {\bfseries 8} (2021) 320--335}.
  

\bibitem{Kels:2020zjn}
A.~P. Kels, ``Interaction-round-a-face and
  consistency-around-a-face-centered-cube,''
  \href{http://dx.doi.org/10.1063/5.0024630}{{\em J. Math. Phys.} {\bfseries
  62} no.~3, (2021) 033509}.

\bibitem{KelsLax}
A.~P. Kels, ``{Lax matrices for lattice equations which satisfy
  consistency-around-a-face-centered-cube},''
  \href{http://dx.doi.org/10.1088/1361-6544/ac1f76}{{\em Nonlinearity}
  {\bfseries 34} no.~10, (2021) 7064--7094}.

\bibitem{NijhoffQ4Lax}
F.~W. Nijhoff, ``Lax pair for the {A}dler (lattice {K}richever-{N}ovikov)
  system,'' \href{http://dx.doi.org/10.1016/S0375-9601(02)00287-6}{{\em Phys.
  Lett. A} {\bfseries 297} no.~1-2, (2002) 49--58}.

\bibitem{nijhoffwalker}
F.~W. Nijhoff and A.~J. Walker, ``The discrete and continuous {P}ainlev\'e {VI}
  hierarchy and the {G}arnier systems,''
  \href{http://dx.doi.org/10.1017/S0017089501000106}{{\em Glasgow Math. J.}
  {\bfseries 43} no.~A, (2001) 109–--123}.
  
\bibitem{Suris_DiscreteTimeToda}
Y.~B. Suris, ``Discrete time {T}oda systems,''
  \href{http://dx.doi.org/10.1088/1751-8121/aacbdc}{{\em J. Phys. A: Math.
  Theor.} {\bfseries 51} (2018) 333001}.

\bibitem{Tremblay2001}
S.~Tremblay, B.~Grammaticos, and A.~Ramani, ``Integrable lattice equations and
  their growth properties,''
  \href{http://dx.doi.org/10.1016/S0375-9601(00)00806-9}{{\em Phys. Lett. A}
  {\bfseries 278} no.~6, (2001) 319--324}.

\bibitem{Viallet2006}
C.-M. Viallet, ``Algebraic entropy for lattice equations,''
  \href{http://arxiv.org/abs/06090430}{{\ttfamily arXiv:06090430 [math-ph]}}.

\end{thebibliography}
\end{document}